\documentclass[11pt,A4paper,fullpage]{article}

\usepackage{amsfonts}
\usepackage{fullpage}
\usepackage{graphicx,color}
\usepackage{xspace}
\usepackage{enumerate}
\usepackage{amssymb,amsmath}
\usepackage{amsthm}
\usepackage{subfig}

\usepackage[T1]{fontenc}
\usepackage[utf8]{inputenc}
\usepackage{authblk}
\usepackage{hyperref}

\newtheorem{theorem}{Theorem}

\newtheorem{corollary}[theorem]{Corollary}
\newtheorem{proposition}[theorem]{Proposition}
\newtheorem{lemma}[theorem]{Lemma}

\newtheorem{property}[theorem]{Property}

\newtheorem{definition}[theorem]{Definition}

\newcommand{\rural}{{\sc Rpt}(r)\xspace}
\newcommand{\china}{{\sc Cpt}(r)\xspace}
\newcommand{\cpt}{{\sc CPT}}
\newcommand{\rpt}{{\sc RPT}}
\newcommand{\comp}{\ensuremath{{\tt{cr}}}\xspace}

\usepackage{xcolor}
\usepackage{amsmath,amssymb,amsthm,xspace}
\usepackage{mathtools}
\usepackage[linesnumbered,ruled]{algorithm2e}
\usepackage{diagbox}

\begin{document}

\title{Online Search with Maximum Clearance}

\author[1]{Spyros Angelopoulos}
\author[2]{Malachi Voss}

\date{}

\affil[1]{Sorbonne Université, CNRS, Laboratoire d’informatique de Paris 6, LIP6, F-75252 Paris, France.
{email: {\tt spyros.angelopoulos@lip6.fr}}
}

\affil[2]{ENS and Sorbonne Université, CNRS, Laboratoire d’informatique de Paris 6, LIP6, F-75252 Paris, France.
{email: {\tt malachi.voss@ens.fr}}
}

\maketitle

\begin{abstract}
We study the setting in which a mobile agent must locate a hidden target in a bounded or unbounded {\em environment}, with no information about the hider's position. In particular, we consider {\em online} search, in which the performance of the search strategy is evaluated by its worst case {\em competitive ratio}. We introduce a multi-criteria search problem in which the searcher has a {\em budget} on its allotted search time, and the objective is to design strategies that are competitively efficient, respect the budget, and maximize the total searched ground. We give analytically optimal strategies for the line and the star environments, and efficient heuristics for general networks.  
\end{abstract}

\section{Introduction}
\label{sec:introduction}

We study a general search problem, in which a mobile agent with unit speed seeks to locate a target that hides in some unknown position of the environment. Specifically, we are given an environment which may be bounded or unbounded, with a point 
$O$ designated as its {\em root}. There is an immobile {\em target} (or {\em hider}) $H$ that is hiding in some unknown point in the environment, whereas the searcher is initially placed at the root $O$. 
The searcher has no information concerning the hider's position.
A {\em search strategy} $S$ determines the precise way in which the searcher explores the environment, and we assume deterministic strategies. The {\em cost}
of $S$ given hider $H$, denoted by $d(S,H)$, is the total distance traversed by the searcher the first time it reaches the location of $H$, or equivalently the total search time.

There is a natural way to evaluate the performance of the search strategy that goes back to~\cite{bellman} and~\cite{beck:yet.more}: we can compare the cost paid by the searcher in a worst-case scenario to the cost paid in the ideal situation where the searcher knows the hider's position. We define the {\em competitive
ratio} of strategy $S$ as
\begin{equation}
\comp(S)=\sup_H \frac{d(S,H)}{d(H)},
\label{eq:comp.ratio}
\end{equation}
with $d(H)$ the distance of $H$ from $O$ in the environment.

Competitive analysis allows to evaluate a search strategy under a status of complete uncertainty, and provides strict, worst-case guarantees.
Competitive analysis has been applied to several search problems in robotics,
for example~\cite{plume},~\cite{magid2004cautiousbug},~\cite{taylor1998vision}~\cite{isler2003local}. See also
the survey~\cite{ghosh2010online}.

In this work we will study the following classes of environments: First, we consider the problem of searching {\em on the line}, informally known as the {\em cow path} problem~\cite{informed.cows}, in which the environment is the unbounded, infinite line. Next, we consider a generalization of linear search, in which the environment consists of $m$ unbounded {\em rays}, concurrent at $O$; this problem is known as the {\em $m$-ray search} or {\em star search} problem. This environment can model much broader settings in which we seek an intelligent allocation of resources to tasks under uncertainty. Thus, it is a very useful paradigm that arises often in applications such as the design of {\em interruptible} systems based on 
{\em contract} algorithms~\cite{steins,spyros:ijcai15,DBLP:conf/podc/KupavskiiW18}, or {\em pipeline filter ordering}~\cite{Condon:2009:ADA:1497290.1497300}. Last, we consider 
general undirected, edge-weighted graph {\em networks}, 
and a target that can hide anywhere over an edge or a vertex of this graph.  

In some previous work, online search may refer to the setting in which the searcher has no information about the environment or the position of the target. In this work we assume that the searcher knows the environment, but not the precise position of the target. This is in line with some foundational work on competitive analysis of online search algorithms, e.g.~\cite{koutsoupias:fixed}. 

\subsection{Searching with a budget}

Most previous work on competitive analysis of searching has assumed that a target is indeed present, and so the searcher will eventually locate it.
Thus, the only consideration is minimizing the competitive ratio. However, this assumption does not reflect realistic settings.
Consider the example of Search-And-Rescue (SAR) operations: first, it is possible that the search mission 
may fail to locate the missing person, in which case searching should resume from its starting point instead of continuing fruitlessly for an exorbitant amount of time. Second, and more importantly, SAR operations come with logistical constraints, notably in terms of the time alloted to the mission.

To account for such situations, in this work we study online search in the setting where the searcher has a certain {\em budget} $T$, which reflects the total amount of search time that it can afford, and a desired competitive ratio $R$ that the search must attain. If the target is found within this budget, the search is {\em successful}, otherwise it is deemed {\em unsuccessful}. We impose two optimization constraints on the search. First, it must be competitively efficient, i.e., its competitive ratio, 
as expressed by~\eqref{eq:comp.ratio} is at most $R$, whether it succeeds or not. 
Second, if the search is unsuccessful, the search has maximized the total {\em clearance} by time $T$. In the case of the environments we study in this work, the clearance is the measure of the part of the environment that the searcher has explored by time $T$. We call this problem the {\em Maximum Clearance problem with budget $T$ and competitive ratio $R$}, and we denote it by {\sc MaxClear}(R,T).

It should be clear that the competitive ratio and the clearance are in a trade-off relation with respect to any given budget $T$: by reducing the competitive efficiency, one can improve the clearance, and vice versa. Hence, our goal is to find strategies that attain the optimal tradeoff, in a Pareto sense, between these two objectives. 

\subsection{Contributions}
\label{subsec:contributions}

We study Maximum Clearance in three environments: the unbounded line, the unbounded star, and a fixed network. We begin with the line: here we show how to use a linear programming formulation to obtain a Pareto-optimal solution. We also show that the Pareto-optimal strategy has a natural interpretation as the best among two simple strategies.

We then move to the $m$-ray star, which generalizes the line, and which is more challenging. Here, we argue that the intuitive strategies that are optimal for the line are not optimal for the star. We thus need to exploit the structure of the LP formulation, so as to give a Pareto-optimal strategy. We do not require
an LP solver, instead, we show how to compute 
the theoretically optimal strategy efficiently, in time $O(m\log m \log T+ m \log T \log \log T)$. Experimental evaluations confirm the superiority of this optimal strategy over other candidate solutions to the problem. 

Finally, we consider the setting in which the environment consists of a network. Here, there is a complication: we do not known the optimal competitive ratio as, for example, in the star (the problem is NP-hard if the target hides on vertices), and only $O(1)$ approximations of the optimal competitive ratio are known~\cite{ANGELOPOULOS2020781}. Hence, in this context, we define {\sc MaxClear}(R,T) with $R\geq 1$, as the problem of maximizing clearance given budget $T$, while guaranteeing that the strategy is an {\em $R$-approximation} of the optimal competitive ratio. Previous approaches to competitive searching in networks typically involve a combination of a solution to the Chinese Postman Problem (CPP)~\cite{edmonds1973matching} with iterative doubling of the search radius. For our problem, we strengthen this heuristic using the {\em Rural Postman Problem} (RPP)~\cite{frederickson1978approximation}, in which only a subset of the network edges need to be traversed. While RPP has been applied to the problem of online 
coverage in robotics~\cite{RPP:coverage},~\cite{easton2005coverage}, to the best of our knowledge, no previous work on competitive search has addressed its benefits. Although there is no gain on the theoretical competitive ratio, our experimental analysis shows that it has significant benefits over the CPP-based approach. We demonstrate this with experiments using real-world data from the library {\em Transportation Network Test Problems}~\cite{transportation:2002}, which model big cities.

We conclude with some extensions and applications. We first explain how our techniques can be applied to a problem ``dual'' to Maximum Clearance, which we call {\em Earliest Clearance}. 
We also show some implications of our work for {\em contract scheduling} problems. In particular, we explain how our results extend those of~\cite{end-guarantees} for contract scheduling with end guarantees.

\subsection{Other related work}
It has long been known that linear search has optimal competitive ratio 9 \cite{beck:yet.more}, which is achieved by a simple strategy based on iterative doubling. Star search on $m$ rays also has a long history of research, going back to~\cite{gal:minimax} who showed that the optimal  competitive ratio is 
\begin{equation}
R^*_m=1+2\rho_m^*, \ \textrm{where } \rho_m^*=\frac{m^m}{(m-1)^{m-1}},
\label{eq:C*_opt}
\end{equation}
a result that was later rediscovered by computer scientists~\cite{yates:plane}. Star search has been studied from the algorithmic point of view in several settings, such as randomized strategies~\cite{ray:2randomized}; multi-searcher strategies~\cite{alex:robots}; 
searching with an upper bound on the target distance~\cite{HIKL99:fixed.distance,revisiting:esa}; fault-tolerant search~\cite{DBLP:conf/podc/KupavskiiW18}; and probabilistic search~\cite{jaillet:online,informed.cows}. 
For general, edge-weighted networks only $O(1)$-approximation strategies are known~\cite{koutsoupias:fixed,ANGELOPOULOS2020781}.

\section{Preliminaries}
\label{sec:preliminaries}

For the $m$-ray star, we assume the rays are numbered $0, \ldots ,m-1$. A search strategy for the star is defined as $\{(x_i,r_i)\}_{i\geq 1}$, with the semantics that in the $i$-th 
{\em step}, the searcher starts from $O$, visits ray $r_i$ to length $x_i$, then returns to $O$. A {\em cyclic}
strategy is a strategy for which $r_i =i \bmod m$; we will thus often omit the $r_i$'s for such strategies, since they are implied.
We make the standing assumption that the target is hiding at least at unit distance from the root, otherwise there is no strategy of bounded competitive ratio.

A {\em geometric} strategy is a cyclic strategy in which 
$x_i=b^i$, for some $b>1$, which we call the {\em base}. Geometric strategies are important since they often give optimally competitive solutions to search problems on a star. For instance, the optimal competitive ratio $R_m^*$ is achieved by a geometric strategy with base $b=\frac{m}{m-1}$~\cite{gal:minimax}. In general, the competitive ratio of a cyclic strategy with base $b$ 
is equal to $1+2\frac{b^m}{b-1}$~\cite{gal:general}. By applying standard calculus, it follows that, for any given
$R = 1+2\rho \geq R_m^*$,
the geometric strategy with base $b$ is $R$-competitive if and only if $b \in [\zeta_1, \zeta_2]$, where 
$\zeta_i$ are the positive roots of the characteristic polynomial $p(t) = t^m-\rho t + \rho$.

A less known family of strategies for the $m$-ray star is the set of strategies which maximize the searched length at the 
$i$-th step. Formally, we want $x_i$ to be as large as possible, so that the strategy $X=(x_i)$ has competitive ratio $R=1+2\rho$. It turns out that this problem has indeed a solution, and as shown in~\cite{jaillet:online}, the resulting strategy $Z=(z_i)$ is one in which the search lengths are defined by the linear recurrence relation $z_{m+i}=\rho(z_{i+1}-z_i)$.
\cite{jaillet:online} give a solution to the recurrence for $\rho=\rho_m^*$. We can show that $Z$ is in fact uniquely defined for all values of $R\geq R^*_m$, and give a closed-form expression for $z_i$, as a function of $\zeta_1$ and $\zeta_2$, defined above (Appendix).
Following the terminology of~\cite{DBLP:conf/stacs/0001DJ19} we call $Z$ the {\em aggressive strategy of competitive ratio $R$}, or simply the {\em aggressive} strategy when $R$ is implied.

For the star we will use a family of linear inequalities involving the search lengths $x_i$ to model the requirement that the search is $R$-competitive. Such inequalities are often used in competitive search, see e.g.~\cite{ultimate}, \cite{HIKL99:fixed.distance}. Each inequality comes from an adversarial position of the target: for a search strategy of the form $X=\{(x_i,r_i)\}$ in the star, the placements of the target which maximize the competitive ratio are on ray $r_j$ and at distance $x_j+\epsilon$, for all $j$ and for infinitesimally small $\epsilon$ (i.e., the searcher barely misses the target at step $j$).

There is, however, a subtlety in enforcing competitiveness in our problem. In particular, we need to filter out some strategies that can be $R$-competitive up to time $T$, but are artificial. 
To illustrate this, consider the case of the line, and a strategy $S$ that walks only to the right of $O$ 
up to time $T$ (it helps to think of $T$ as very large). This strategy is 1-competitive in the time interval $[0,T]$, and obviously maximizes clearance, but intuitively is not a realistic solution.
The reason for this is that $S$ discards the entire left side with respect to $R$-competitiveness. Specifically, for a point at distance 1 to the left of $O$, 
any extension $S'$ of $S$ will incur a competitive ratio of at least $2T+1$, which can be enormous.

We thus need to enforce a property that intuitively states that a feasible strategy $S$ to our problem should be extendable to an $R$-competitive strategy $S'$ that can detect targets hiding infinitesimally beyond the boundary that has been explored by time $T$ in $S$.   
We call this property {\em extendability} of an $R$-competitive strategy. We give a formal definition in the Appendix concerning our environments, although this intuitive description will suffice for the purposes of modeling and analysis.  Our experimental evaluation shows that the optimal extendable strategy on the star performs significantly better than other candidate strategies, which further justifies the use of this notion.

\section{A warm-up: Maximum Clearance on the line}
\label{sec:line}

We begin with the simplest environment: an unbounded line with root $O$. Fix a competitive ratio $R=1+2\rho$, for some $\rho\geq \rho_2^*=4$. Without loss of generality, we assume cyclic strategies $X=(x_i)$ such that $x_{i+2} >x_{i}$, for all $i$.

Let ${\cal S}_k$ denote the set of all strategies $X=(x_1,\ldots x_k)$ with $k$ steps.
We can formulate {\sc MaxClear}(R,T) {\em restricted} to ${\cal S}_k$ using the following LP, which we denote $L_2^{(k)}$.

\begin{align*}
\text{max} \quad & x_{k-1} + x_k  & \tag{$L_2^{(k)}$}\\
\text{subject to} \quad 
&x_1\leq \rho  \tag{$C_0$}\\
&\sum\nolimits_{i=1}^{j+1} x_{i} \leq \rho \cdot x_j, &j \in [1, k-2] \tag{$C_j$}\\
&\sum\nolimits_{i=1}^{k} x_{i}  \leq \rho \cdot x_{k-1} \tag{$E_{k-1}$} \\
&2 \sum\nolimits_{i=1}^{k-1} x_{i}  + x_k \leq T      \tag{$B$} 
\end{align*}
In this LP, constraints $(C_0)$ and $(C_1), \ldots (C_{k-2})$ model the requirement for 
$(1+2\rho)$-competitiveness. $(C_0)$ models a target hiding at distance 1 from $O$, whereas the remaining constraints model a target hiding right after the turn points of $x_1, \ldots x_{k-2}$, respectively. 
Constraint $(B)$ is the budget constraint. Last, constraint $(E_{k-1})$ models the extendability property, which on the line means remaining competitive for a target hiding just beyond the turn point of $x_{k-1}$.

Therefore, an optimal strategy is one of maximum objective value, among all feasible solutions to $L_2^{(k)}$, for 
all $k\geq 1$. We will use this formulation to show that the optimal strategy has an intuitive statement. 
Let $Z=(z_i)$ be the aggressive strategy of competitive ratio $R$. 
From $Z$ we derive the {\em aggressive strategy with budget $T$}, which is simply the maximal prefix of $Z$ that satisfies the budget constraint $(B)$. We denote this strategy by $Z_T$.

Note that $Z_T$ may be wasteful, leaving a large portion of the budget unused, which suggests another intuitive strategy derived from $Z$. Informally, one can ``shrink'' the search lengths of $Z$ in order to deplete the budget precisely at some turn point. Formally, we define the {\em scaled aggressive strategy with budget $T$}, denoted by $\tilde{Z}_T$ as follows.
Let $l$ be the minimum index such that $2\sum_{i=1}^{l-1}z_i+z_l \geq T$, and define $\gamma$ as $T/(2\sum_{i=1}^{l-1}z_i+z_l)$. Then $\tilde{Z}_T$ is defined as $(\tilde{z}_i)=(\gamma \cdot z_i)$.

We will prove that one of $Z_T$, and $\tilde{Z}_T$ is the optimal strategy. We can first argue about constraint tightness in  an optimal solution to $L_2^{(k)}$.

\begin{lemma}
In any optimal solution to $L_2^{(k)}$, at least one of the constraints $(C_0)$ and $(B)$ is tight, and all other constraints must be tight. 
\label{lemma:line.one.constraint}
\end{lemma}
\begin{proof}
By way of contradiction, let $X^*=(x_i^*)$ denote an optimal solution for the LP which does not obey the conditions of the lemma. Recall that we only consider solutions on the line which explore strictly farther each time they visit a side, i.e. $x_{i+2}>x_i$.

Suppose that a constraint $(C_j)$ is loose.
Then we could decrease $x_j^*$ by a small amount, say $\delta$, and increase $x_k^*$ by $\delta$, maintaining feasibility, including the implicit constraint $x^*_{j-2}<x^*_j$, and improving the objective, a contradiction.

Similarly, if $(E_{k-1})$ is not tight, then we could decrease $x_{k-1}^*$ by a small amount, say $\delta$, and increase $x_{k}^*$ by $2\delta$, maintaining feasibility, including the implicit constraint $x^*_{k-3}<x^*_{k-1}$, and improving the objective, a contradiction.

It remains then to argue that one of the constraints $(C_0)$ and $(B)$ is tight. This is true because if they are both slack, then 
there would exist $\alpha>1$ such that $(\alpha \cdot x_i^*)$ is a feasible solution with a better objective value than $X^*$, a contradiction.
\end{proof}

Lemma~\ref{lemma:line.one.constraint} shows that if $X^*$ is optimal for $L_2^{(k)}$, then one can subtract successive constraints from each other to obtain the linear recurrence relation $x^*_{i+2} = \rho(x^*_{i+1}-x^*_i)$, with constraint $(C_1)$ giving an initial condition. So $X^*$, viewed as a point in $\mathbb{R}^k$, is on a line $\Delta \subset \mathbb{R}^k$, defined as the set of all points which satisfy $(C_1),\dots,(E_{k-1})$ with equality. This leaves us with two possibilities: either $X^*=X_0^{(k)}$ the point on $\Delta$ for which $(C_0)$ is tight, or $X^*=X_B^{(k)}$ the point on $\Delta$ for which $(B)$ is tight.

Define now $\mathcal{X}_0$ as the set of all feasible points $X_0^{(k)}$ and $\mathcal{X}_B$ as the set of all feasible points $X_B^{(k)}$. A point $X$ is {\em optimal} for one of these sets if its objective value is no worse than any point in that set. The following lemma is easy to see for $Z_T$, and requires a little more effort for $\tilde{Z}_T$.

\begin{lemma}
$Z_T$ is optimal for $\mathcal{X}_0$, and $\tilde{Z}_T$ is optimal for $\mathcal{X}_B$.
\label{lemma:line.one.for.each}
\end{lemma}

\begin{proof}
$X_0^{(k)}$ is simply a prefix $(z_1,\dots,z_k)$ of the aggressive strategy $Z$, because the formulas defining them are identical. Because $z_i$ is increasing (see the formulas for $Z$ given above), the objective value of $X_0^{(k)}$ is increasing, and so $Z_T$, which is the longest feasible prefix for $L_2^{(k)}$, is optimal for $\mathcal{X}_0$.

$X_B^{(k)}$ is a scaled version of $X_0^{(k)}$ (they both belong to the same line $\Delta$), and so is given by $(\gamma_k z_i)$ where $\gamma_k = \frac{T}{2 S_{k-1}^{(Z)} + z_k}$. Denote ${\bf clr}(X)$ the objective value, or clearance, of a strategy $X$: we have ${\bf clr}(X_B^{(k)}) =\gamma_k (z_{k-1}+z_k)= \frac{z_{k-1}+z_k}{2 \rho z_{k-1}^{(Z)}-z_k}T$, using the identity 
$(E_{k-1})$, which holds because $X_B^{(k)}\in \Delta$.
We want to show that the clearance of $(X_B^{(k)})$ decreases with $k$. A short calculation yields:
\[{\bf clr}(X_B^{(k)})\leq {\bf clr}(X_B^{(k-1)}) \Leftrightarrow \frac{z_k}{z_{k-1}} \leq \frac{z_{k-1}}{z_{k-2}}. \]

We now make use of the formulas for $z_i$. For optimal $\rho=\rho_m^*$, we get
\[ \frac{z_{i+1}}{z_i} = \frac{m+i}{m+i-1}\cdot\frac{m}{m-1}, \]
which is indeed decreasing, and for $\rho>\rho_2^*$ a short calculation yields
\[ \frac{z_{i+1}}{z_i}\leq \frac{z_i}{z_{i-1}} \Leftrightarrow \zeta_1^2 + \zeta_2^2 \geq 2\zeta_1\zeta_2.\]
Therefore for all $\rho$, $\tilde{Z}_T$ is optimal for $\mathcal{X}_B$.
\end{proof}

From Lemma~\ref{lemma:line.one.constraint} and~\ref{lemma:line.one.for.each} we conclude that the better of the two strategies $Z_T$ and $\tilde{Z}_T$ is optimal for {\sc Max}(R,T) on the line. We call this strategy the {\em mixed aggressive} strategy.

\section{Maximum Clearance on the Star}
\label{sec:star}

We now move to the $m$-ray star domain. We require that the strategy be $(1+2\rho)$-competitive, for some given $\rho\geq \rho^*_m$, where $\rho^*_m=\frac{m^m}{(m-1)^{m-1}}$, and we
are given a time budget $T$.

\subsection{A first, but suboptimal approach}
\label{subsec:star.other}

An obvious first place to look is the space of geometric strategies. We wish the geometric strategy to have competitive ratio $1+2\rho$, so the strategy must have base $b \in [\zeta_1,\zeta_2]$, using the notation of the preliminaries. Since we want to maximize the clearance of our strategy, it makes sense to take $b=\zeta_2$. We define the {\em scaled geometric strategy with budget T} similarly to the scaled aggressive strategy: find the first step at which the budget $T$ is depleted, and scale down the geometric strategy so that it depletes $T$ precisely at the end of that step. The scaled geometric strategy represents the best known strategy prior to this work, but is suboptimal.

For Maximum Clearance on the line, we proved that the optimal strategy is the best of the aggressive and the scaled aggressive strategies. One may ask then whether the optimal strategy in the star domain can also be expressed simply as the better of these two strategies. The answer is negative, as we show in the experimental evaluation.

\subsection{Modeling as an LP}

As with the line, we first show how to formulate the problem using a family of LPs, denoted by $L^m_k$, partitioning strategies according to their length $k$. For a given step $j$, we denote by $\bar{\jmath}$  the previous step for which the searcher visited the same ray, i.e, the maximum $\bar{\jmath} <j$
such that $r_{\bar{\jmath}}=r_j$, assuming it exists, otherwise we set $x_{\bar{\jmath}}=1$.
We denote by $l_r$ the last step
at which the searcher explores ray $r$. Finally, we denote by $j_0$ the last step in which the searcher searches a yet unexplored ray, i.e., the largest step $j$ such that $\bar{\jmath}=0$. This gives us:

\begin{align*}
\text{max} &\quad  \sum \nolimits_{i=1}^m x_{l_i}  & \tag{$L_m^{(k)}$}\\
\text{subject to}& \quad  \sum \nolimits_{i=1}^{j_0} x_i \leq \rho &\tag{$C_0$} \\ 
&\sum \nolimits_{i=1}^{j-1} x_i \leq \rho \cdot x_{\bar{\jmath}},	\qquad \;	j \in [j_0+1, k] &\tag{$C_j$}\\
&\sum \nolimits_{i=1}^k x_i \leq \rho \cdot x_{l_r},\;	  r \in [1,m], l_r\neq r_k &\tag{$E_r$} \\
&2 \sum \nolimits_{i=1}^{k-1} x_i + x_k \leq T    &\tag{$B$}
\end{align*}

Here, constraints $(C_0), (C_{j_0}), \ldots (C_k)$ model the $(1+2\rho)$-competitiveness of the strategy, and constraint $(B)$ models the budget constraint.
Constraints $(E_1), \ldots ,(E_m)$ model the extendability property, by giving competitiveness constraints for targets placed just beyond the turn points  at $x_{l_1},\dots,x_{l_r}$. Details concerning the derivation of all constraints can be found in the Appendix.

As is standard in star search problems, we can add some much-needed structure in the above formulation.
\begin{theorem}[Appendix]
Any optimal solution $X^*=(x^*_i,r_i)$ to $L_m^{(k)}$ must be {\em monotone} and {\em cyclic}: $(x^*_i)$ is increasing and $r_i = i \mod m$ up to a permutation.
\label{theorem:cyclic}
\end{theorem}
This means that we can formulate the problem using a much simpler family of LPs which we denote by $P_m^{(k)}$, where constraints $(M_i)$ model monotonicity.

\begin{align*}
\text{max}& \quad  \sum \nolimits_{i=0}^{m-1} x_{k-i}   &\tag{$P_m^{(k)}$}\\
\text{subj to}& \quad \sum \nolimits_{i=1}^{m-1} x_i\leq \rho &\tag{$C_0$} \\ 
&\sum \nolimits_{i=1}^{j+m-1} \!x_i\leq \rho \cdot x_{j},   \qquad\,	  j \in [1, k-m] &\tag{$C_j$}\\
&\!\!\!\!\sum \nolimits_{i=1}^k x_i \leq \rho\cdot x_j,  	\;\;	j \in [k-m+1,k-1]  &\tag{$E_j$} \\
&x_i \leq x_{i+1},	\qquad\qquad\qquad\quad\;	  i \in [1,k-1]&\tag{$M_i$}\\
&2\sum \nolimits_{i=1}^{k-1} x_i+ x_k \leq T    &\tag{$B$}
\end{align*}

\subsection{Solving $P_m^{(k)}$}
\label{subsec:star.optimal}

While proving cyclicality, we also prove that for any optimal solution to $L_m^{(k)}$, most of the constraints are tight, similarly to Lemma~\ref{lemma:line.one.constraint}. Applying this result to $P_m^{(k)}$ gives the following.
\begin{lemma}
In an optimal solution to the LP $P_m^{(k)}$, constraints $(M_i)$ are not necessarily tight, at least one of the constraints $(C_0)$ and $(B)$ is tight, and all other constraints must be tight.
\label{lemma:star.one.constraint}
\end{lemma}

Subtracting $(C_i)$ from $(C_{i+1})$ and $(C_{k-m})$ from $(E_{k-m+1})$ gives a linear recurrence formula which any optimal solution $X^*$ must satisfy:
\[ x_{i+m}^* = \rho(x_{i+1}^*-x_i^*). \quad i \in [1,k-m] \]
The constraints $(E_j)$ give us $m-1$ equations to help determine the solution: $\rho x_{k-m+1}^*=\dots=\rho x_{k-1}^*=S_k$. So $X^*$, viewed as a point in $\mathbb{R}^k$, is on a line $\Delta_m^{(k)}\subset \mathbb{R}^k$, defined as the set of all points which satisfy $(C_1),\dots,(E_{k-1})$ with equality.
Lemma \ref{lemma:star.one.constraint} shows that the solution to $P_m^{(k)}$ is either the point $X_0^{(k)} \in \Delta_k^m$ for which constraint $(C_0)$ is tight, or the point $X_B^{(k)} \in \Delta_k^m$ for which constraint $(B)$ is tight.

We can compute these two strategies efficiently for a fixed $k$, as we will demonstrate for $X_B^{(k)}$.
We rewrite the conditions $X_B^{(k)}\in \Delta_k^m$ and ``$(B)$ is tight'' as a matrix equation:
\begin{equation}
{\cal M}_{k,B}^m \times X
=
\begin{pmatrix}
0 & \cdots & 0 & T \\
\end{pmatrix}^\top
\label{equation:the.matrix}
\end{equation}
where ${\cal M}_{k,B}^m$ is the following $k\times k$ matrix:

\[
\begin{pmatrix}
\rho	&-\rho	&0	&0	&\cdots &1	&0	&\cdots &0	&0	&0 \\
0	&\rho	&-\rho	&0	&\cdots &0	&1	&\cdots &0	&0 	&0 \\
\vdots 	&\vdots &\vdots &\vdots&\ddots &\vdots&\vdots &\ddots &\vdots &\vdots&\vdots \\
0	&0	&0	&0	&\cdots &0	&0	&\cdots	&\rho	&-\rho	&0 \\
1	&1	&1	&1	&\cdots &1	&1	&\cdots	&1	&1-\rho	&1 \\
2	&2	&2	&2	&\cdots &2	&2	&\cdots	&2	&2	&1 \\
\end{pmatrix}
\]

${\cal M}_{k,B}^m$ has a very nice structure, and is very sparse, as all coefficients are concentrated in three diagonals (numbered $1$, $2$, and $m+1$) and the last two lines. This is good for us: we can solve \eqref{equation:the.matrix} in time $O(k)$ using Gaussian elimination. $X_0^{(k)}$ can be computed similarly, using the matrix ${\cal M}_{k,0}^m$, which is identical to ${\cal M}_{k,B}^m$ except for the last line, which contains $(C_0)$, and \eqref{equation:the.matrix} becomes ${\cal M}_{k,0}^m \times X_0^{(k)} = (0 \cdots 0 \, \rho)^\top$. When solving \eqref{equation:the.matrix} we discarded the constraint $(C_0)$, so we need to check whether $X_B^{(k)}$ is feasible for this constraint. Similarly, we need to check whether $X_0^{(k)}$ is feasible for $(B)$.

\subsection{Finding the optimal strategy}

At this point, we have determined how to compute two families of strategies, the sets $\mathcal{X}_0=\{ X_0^{(k)}, k\in \mathbb{N}\}$ and $\mathcal{X}_B = \{X_B^{(k)},k\in\mathbb{N}\}$, and we have shown that any optimal strategy belongs to one of these two families. Define $k_0$ the highest $k$ for which $X_0^{(k)}$ is feasible, and $k_B$ the lowest $k$ for which $X_B^{(k)}$ is feasible. We conclude with our two main results.
\begin{theorem}[Appendix]
$X_0^{(k)}$ is feasible if and only if $k\leq k_0$, and $X_B^{(k)}$ is feasible if and only if $k\geq k_B$. Moreover, $X_0^{(k_0)}$ is optimal for $\mathcal{X}_0$, and $X_B^{(k_B)}$ is optimal for $\mathcal{X}_B$.
\label{theorem:star.optimal}
\end{theorem}
\noindent \textit{Proof sketch.} We show first that any point $(x_i)$ that is feasible for $P_m^{(k)}$ is positive: $\forall i,x_i\geq 0$. Denote $X_0^{(k)}=(x_i)$ and $X_0^{(k-1)}=(y_i)$. Using the convention $y_0=1$, the strategy $D=(x_i-y_{i-1})$ is feasible for $P_m^{k}$, therefore positive. This means that $X$ has a higher objective value than $Y$, and also requires a larger budget: this shows that $k_0$ is well-defined and optimal. Because $X_0^{(k)}$ and $X_B^{(k)}$ are scaled versions of each other, we get $k_B = k_0$ or $k_0+1$. Additional calculations show that the objective values of $X_B^{(k)}$ are decreasing.

\begin{theorem}
The optimal strategy for the $m$-ray star can be computed in time $O(m\log(T)\log(m\log(T)))$.
\label{thm:complexity.star}
\end{theorem}
\noindent \textit{Proof sketch.} The scaled geometric strategy with base $b=\frac{m}{m-1}$ is a feasible point for a certain $P_m^{(k_G)}$, with $k_G=O(\log_{b}(T))=O(m\log(T))$. This means that $X_B^{(k_G)}$ is feasible, and so $k_B\leq k_G$ gives us an upper bound. We can use binary search to find $k_B$, solving \eqref{equation:the.matrix} at each step at a cost of $O(k_G)$. We know that $k_0$ is either $k_B$ or $k_B-1$, so all that remains is to compare the two strategies, which gives us a total complexity of $O(m\log(T)\log(m\log(T)))$.

\section{Maximum Clearance in a Network}
\label{sec:network}

In this section we study the setting in which the environment is a {\em network}, represented by an undirected, 
edge-weighted graph $Q=(V,E)$, with a vertex $O$ designated as the root. Every edge has a non-negative {\em length} which represents the distance of the vertices incident to the edge. The target can hide anywhere along an edge, which means that the search strategy must be a traversal of all edges in the graph. We can think of the network 
$Q$ as being endowed with Lebesgue measure corresponding to the length. This allows as to define, for a given subset $A$ of the network, its measure $l(A)$. Informally, $l(A)$ is the total length of all edges (partial or not) that belong in $A$. Given a strategy $S$ and a target $t$, the cost $d(S,t)$ and the distance $d(t)$ are well defined, and so is the competitive ratio according to~\eqref{eq:comp.ratio}. We will denote by $Q[r]$ the subnetwork that consists of all points in $Q$ within distance at most $r$ from $O$.

The exact competitive ratio of searching in a network is not known, and there are only
$O(1)$-approximations~\cite{koutsoupias:fixed, ANGELOPOULOS2020781} of the optimal competitive ratio.  
For this reason, as explained in the introduction, we interpret {\sc MaxClear(R,T)} as a maximum clearance strategy with budget $T$ that is an {\em $R$-approximation} of the optimal competitive ratio.
The known approximations use searching based on iterative deepening, e.g. strategy \china, which in each round $i$, searches $Q[r^i]$ using a Chinese Postman Tour (CPT)~\cite{edmonds1973matching} of $Q[r^i]$, for some suitably chosen value of $r$. 

We could apply a similar heuristic to the problem of Maximum Clearance. However, searching using a CPT of $Q[r^i]$
is wasteful, since we repeatedly search parts of the network that have been explored in rounds
$1 \ldots i-1$. Instead, we rely on heuristics for the {\em Rural Postman Problem}~\cite{frederickson1978approximation}. In this problem, given an edge-weighted network $Q=(V,E)$, and a subset $E_{\text{req}} \subseteq E$ of {\em required} edges, the objective is to find a minimum-cost traversal of all edges in 
$E_{\text{req}}$in $Q$; we call this tour RPT for brevity.
Unlike the Chinese Postman Problem (CPP), finding an RPT is NP-hard. The best known 
approximation ratio is 1.5~\cite{frederickson1978approximation}, but several heuristics have been proposed, e.g.~\cite{corberan2010recent},~\cite{hertz1999improvement}.

We thus propose the following strategy, which we call \rural. For each round $i \geq 1$, let $R_{i-1} = Q[r^i]\setminus Q[r^{i-1}]$ denote the part of the network that the searcher has not yet explored in the beginning of round $i$ (and needs to be explored). Compute both tours 
\cpt$(Q[r^i])$ and \rpt$(Q[r^i])$, the latter with required set of edges the edge set of $R_{i-1}$ (using the 1.5-approximation algorithm),  and choose the tour of minimum cost among them. This continues until the time budget $T$ is exhausted.  It is very hard to argue from a theoretical standpoint that the use of RPT yields an improvement on the competitive ratio; nevertheless, the experimental evaluation shows that this is indeed beneficial to both competitiveness and clearance. Since \rural is at least as good as a strategy that is purely based on CPTs, we can easily show the following, which is proven analogously to the randomized strategies of~\cite{ANGELOPOULOS2020781}.

\begin{proposition}
For every $r>1$, \rural is a $\frac{r^2}{r-1}$-approximation of the optimal competitive ratio. 
In particular, for $r=2$, it is a 4-approximation.
\label{prop:rural}
\end{proposition}

\begin{proof}
Let $\ell_i$ denote the length of the optimal CPT in $G[r^i]$. The competitive ratio of the strategy is at most
\[
\sup_{j \geq 1} \frac{\sum_{i=1}^j \ell_i}{r^{j-1}}.
\] 
Let $R^*$ denote the optimal (deterministic) competitive ratio. Then it holds that for every $i$,
$
R^*\geq \frac{l_i}{r^i}.
$ 
This is because any deterministic strategy needs time at least $l_i$ to traverse $G[r^i]$, and every point in
$G[r^i]$ is at distance at most $b^i$ from $O$. Combining the above inequalities, we obtain that the competitive ratio of \rural is at most
\[
R^* \cdot \sup_{j \geq 1} \frac{\sum_{i=1}^j r^i}{r^{j-1}} \leq R^* \cdot \sup_{j\geq 1} 
\frac{r^{j+1}-1}{r^{j-1}(r-1)} \leq R^* \cdot \frac{r^2}{r-1}.
\] 
The last inequality implies that the best approximation factor is achieved for $r=2$, and is equal to 4. 
\end{proof}

Note that \rural is, by its statement, extendable, since it will always proceed to search beyond the boundary of round $i$ in round $i+1$. Moreover, \rural is applicable to unbounded networks as well, provided that for any $D$, the number of points in the network at distance $D$ from $O$ is bounded by a constant. This is necessary for the competitive ratio to be bounded~\cite{ANGELOPOULOS2020781}.

\section{Experimental evaluation}

\subsection{$m$-ray star}
In this section we evaluate the performance of our optimal strategy against two other candidate strategies.
The first candidate strategy is the scaled geometric strategy, with base $\zeta_2, $which we consider as the baseline for this problem prior to this work. The second candidate strategy is the mixed aggressive strategy. 
Recall that we defined both strategies at the beginning of the star section, and that all these strategies are defined for the same competitive ratio $R$.

Figure~\ref{fig:star.gain} depicts the relative performance of the optimal strategy versus the performance of the other two strategies, for $m=4$, and optimal competitive ratio $R=R^*_4$, for a range of budget values $T\in[10,10^{15}]$. Once the budget $T$ becomes meaningfully large (i.e, $T\geq 50$), the optimal strategy dominates the other two, outperforming both by more than $20\%$. In contrast, the mixed aggressive strategy offers little improvement over the scaled geometric strategy for every reasonably large value of $T$.

\begin{figure}[htb!]
\centering
\includegraphics[width=0.65\linewidth]{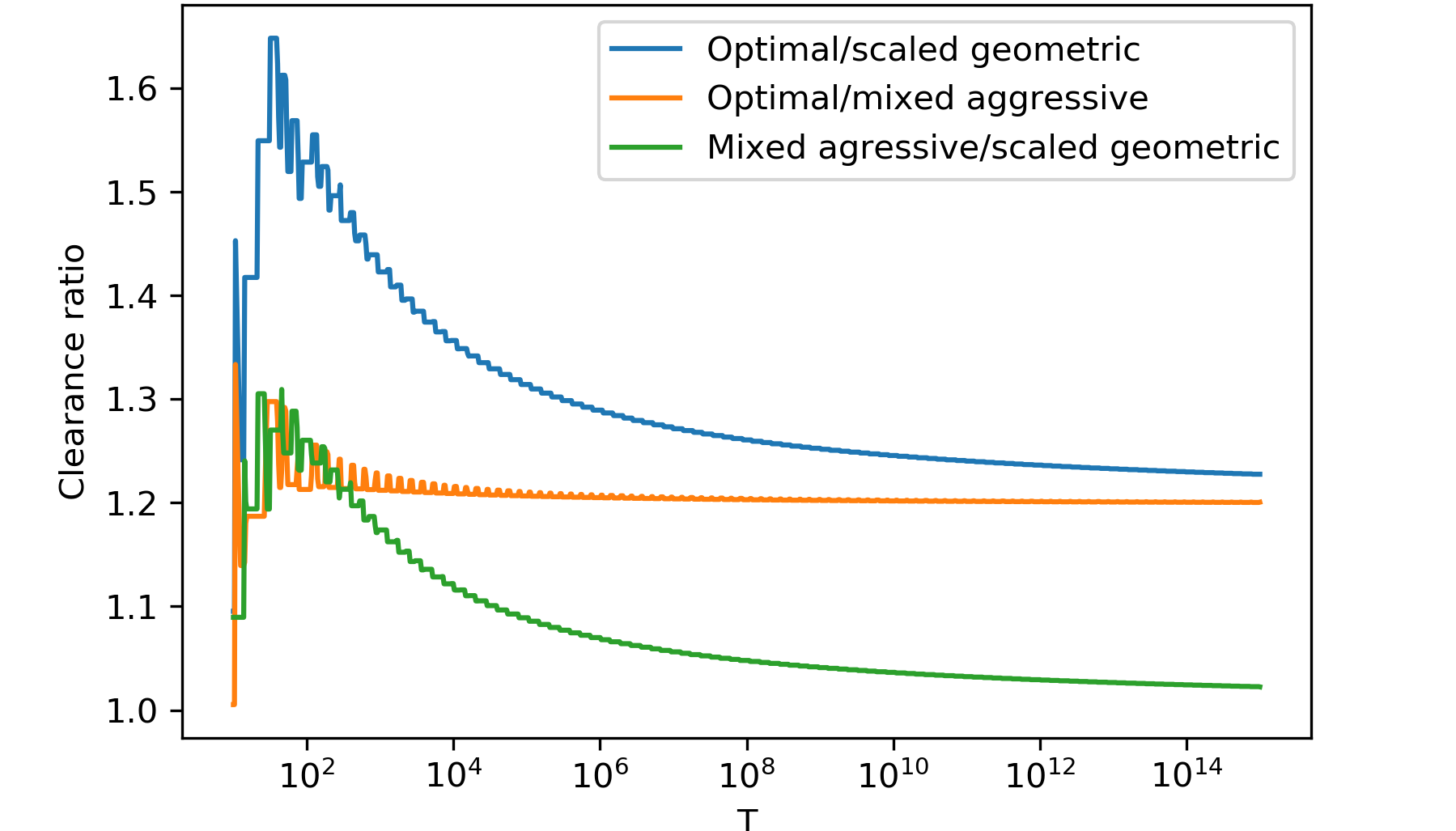}
\caption{Clearance ratios for $m=4$ and $R=R_4^*$, as function of $T$.}
\label{fig:star.gain}
\end{figure}

\begin{figure}[htb!]
\centering
\includegraphics[width=0.65\linewidth]{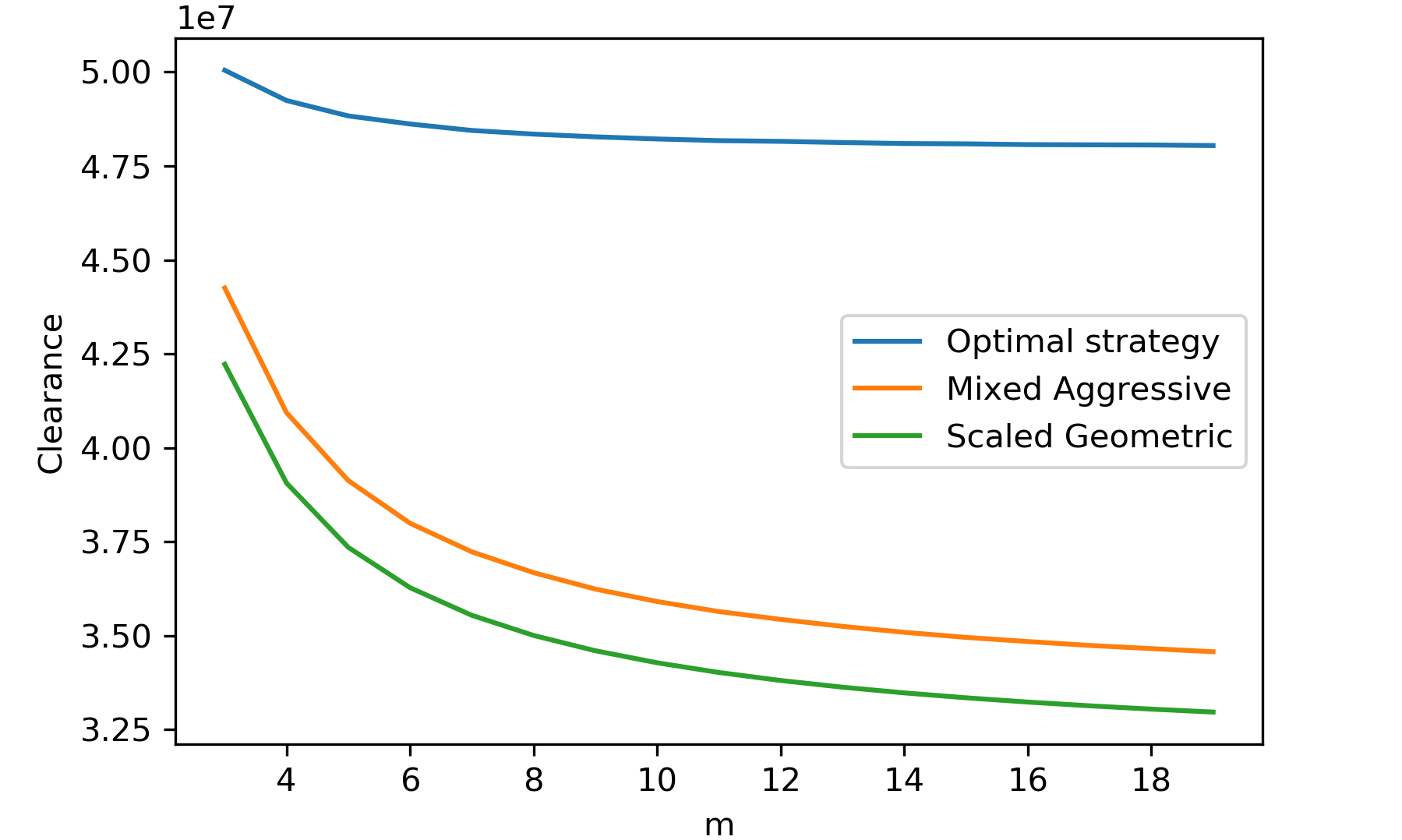}
\caption{Clearance as function of $m$, for $T=10^8$ and $R=R_m^*$.}
\label{fig:star.m}
\end{figure}
Figure \ref{fig:star.m} depicts the influence of the parameter $m$ on the clearance achieved by the three strategies, for a relatively large value of $T=10^8$. For each value of $m$ in $[3,20]$, we require that the strategies have optimal competitive ratio $R=R_m^*$.  
We observe that as $m$ increases, each strategies' clearance decreases, however the optimal strategy is far less impacted. This means that as $m$ increases, the relative performance advantage for the optimal strategy also increases, in comparison to the other two.

Figure \ref{fig:star.rho} depicts the strategies' performance for $m=4$, and $T=10^4$, as a function of the competitive ratio $R\geq R_4^*$. In particular, we consider $R\in [R_4^*, 3R_4^*]$. We observe that as $R$ increases, the mixed aggressive strategy is practically indistinguishable from the scaled geometric.
The optimal strategy has a clear advantage over both strategies for all values of $R$ in that range.

\begin{figure}[htb!]
\centering
\includegraphics[width=0.65\linewidth]{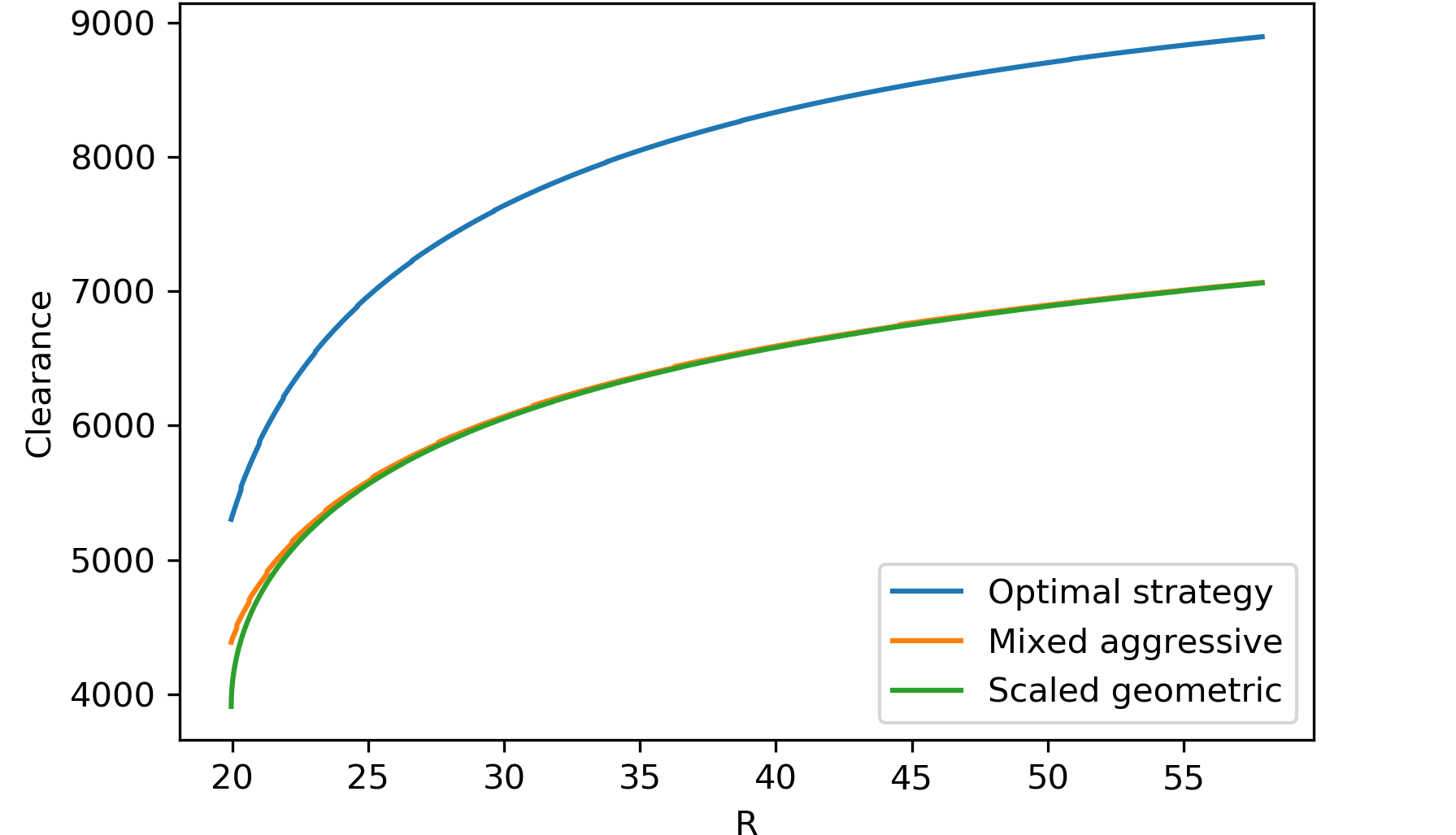}
\caption{Clearance as function of $R$, for $m=4$ and $T=10^4$.}
\label{fig:star.rho}
\end{figure}

More experimental results can be found in the Appendix. 

\subsection{Networks}

We tested the performance of \rural against the performance of \china. Recall that the former searches 
the network $Q[r^i]$ iteratively using the best among the two tours \cpt$(Q[r^i])$ and \rpt$(Q[r^i])$, whereas the latter uses only the tour \cpt($Q[r^i])$.  We found $r=2$ to be the value that optimizes the competitive ratio in practice, as predicted also by Proposition~\ref{prop:rural}, so we chose this value for our experiments. 

We used networks obtained from the online library {\em Transportation Network Test Problems}~\cite{transportation:2002}, after making them undirected. This is a set of benchmarks that is very frequently used in the assessment of transportation network algorithms (see e.g.~\cite{jahn2005system}). 
The size of the networks we chose was limited by the $O(n^3)$ time-complexity of \china and \rural ($n$ is the number of vertices). For \rpt \  we used the algorithm due to~\cite{frederickson1978approximation}.

Figures \ref{fig:berlin} and \ref{fig:chicago} depict the clearance achieved by each heuristic, as function of the budget $T$, for a root chosen uniformly at random. The first network is a European city with no obvious grid structure, whereas the second is an American grid-like city. 
We observe that the clearance of \china exhibits plateaus, which we expect must occur early in each round, since 
\cpt \ must then traverse previously cleared ground. We also note that these plateaus become rapidly larger as the number of rounds increases, as expected. In contrast, \rural entirely avoids this problem, and 
performs significantly better, especially for large time budget.

\begin{figure}[htb!]
\centering
\includegraphics[width=0.65\linewidth]{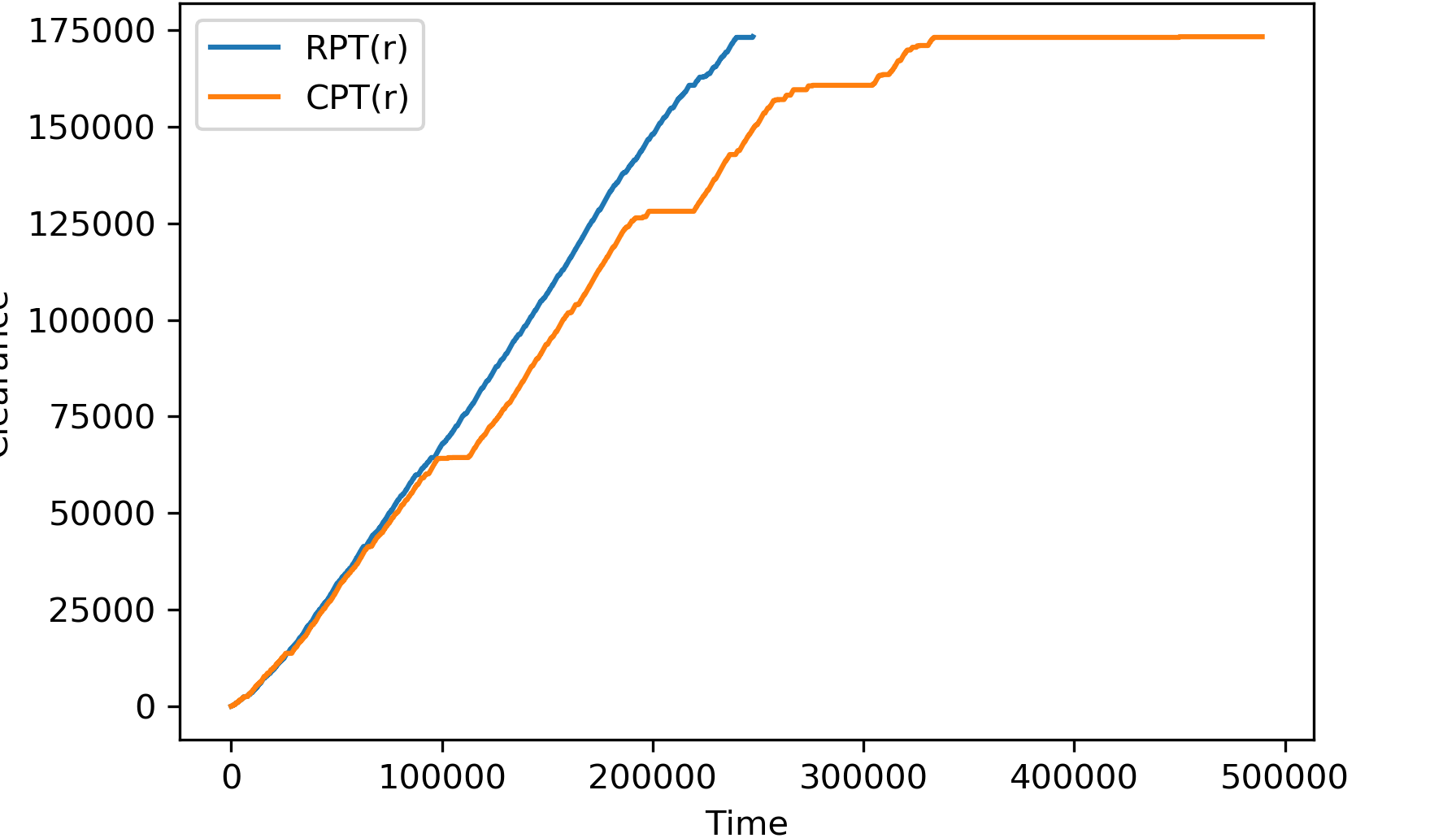}
\caption{Comparison of the two strategies on the Berlin network (633 nodes, 1042 edges).}
\label{fig:berlin}
\end{figure}

\begin{figure}[htb!]
\centering
\includegraphics[width=0.65\linewidth]{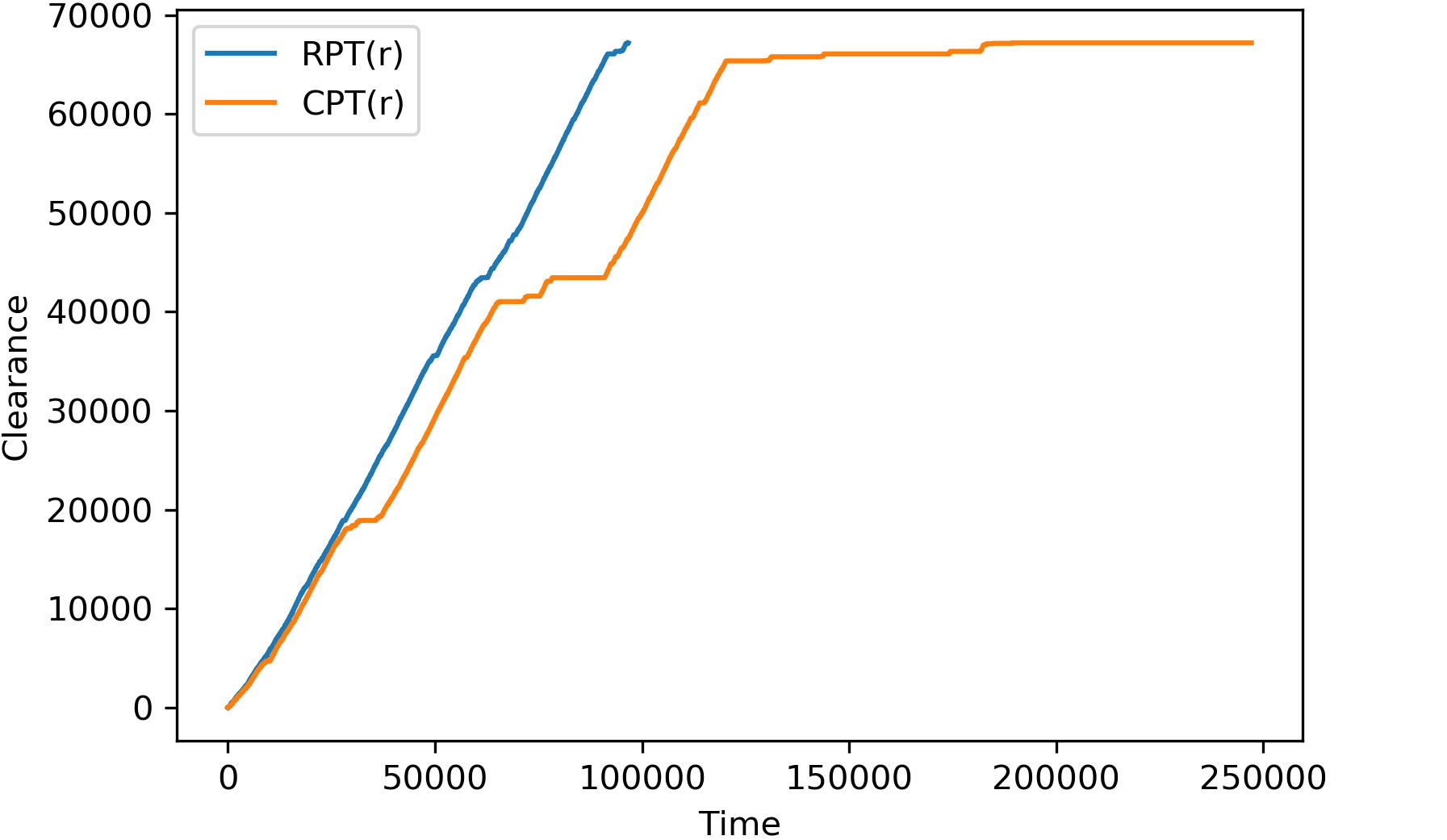}
\caption{Comparison of the two strategies on the Chicago network (933 nodes, 1475 edges).}
\label{fig:chicago}
\end{figure}

Figure \ref{fig:compare} depicts the ratio of the average clearance of \rural over the average clearance of \china  as a function of the time budget $T$, calculated over 10 random runs of each algorithm on the Berlin network (each run with a root chosen uniformly at random).
We observe  that \rural consistently outperforms \china, by at least 8\% for most values of $T$, and up to 16\% when $T$ is comparable to the total length of all edges in the graph (173299). 
At $T=250000$, in most runs, \rural has cleared the entire network.

\begin{figure}[htb!]
\centering
\includegraphics[width=0.65\linewidth]{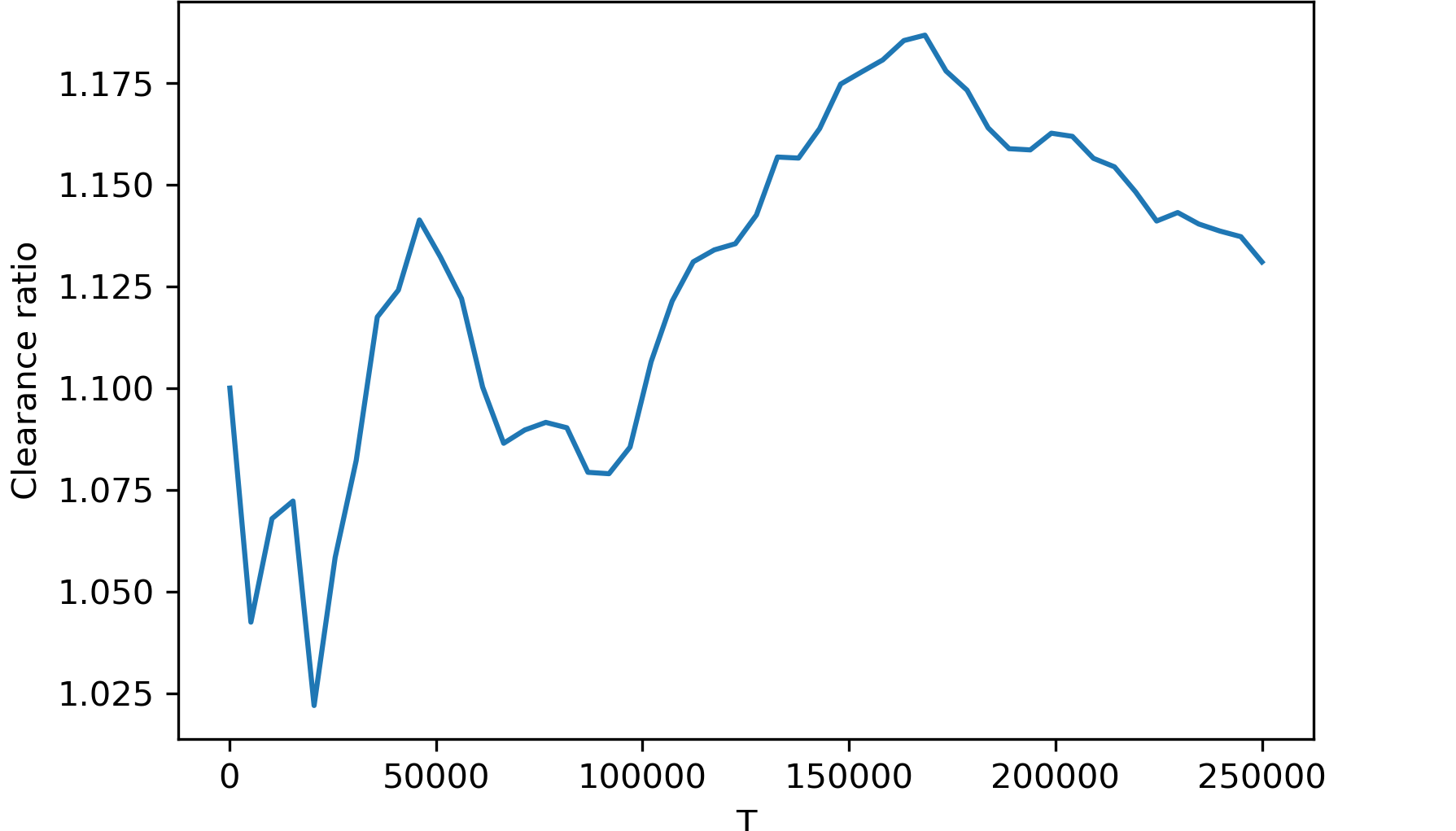}
\caption{Clearance ratio of \rural versus \china, for 10 randomly chosen roots, for the Berlin network.}
\label{fig:compare}
\end{figure}

The average competitive ratios for these runs are $160$ for \china and $132$ for \rural, demonstrating a clear advantage.
More experimental results can be found in the Appendix.

\section{Extensions and conclusions}
\label{sec:extensions}

One can define a problem ``dual'' to Maximum Clearance, which we call Earliest Clearance. Here, we are given a bound $L$ on the desired ground that we would like the searcher to clear, a required competitive ratio $R$, and the objective is to design an $R$-competitive strategy which minimizes the time to attain clearance $L$. The techniques we use for Maximum Clearance can also apply to this problem, in fact Earliest Clearance is a simpler variant; e.g., for star search, optimal strategies suffice to saturate all but one constraint, instead of all but two (see Appendix).

Maximum Clearance on a star has connections to the problem of scheduling {\em contract algorithms with end guarantees}~\cite{end-guarantees}. More precisely, our LP formulation has certain similarities with the formulation used in that work (see the LP $P_m$, on page 5496 in~\cite{end-guarantees}), and both works use the same general approach: first, a technique to solve the LP of index $k$, and then a procedure for finding the optimal index $k^*$. However, there are certain significant differences. First, our formulations allow for any competitive ratio $\rho \geq \rho_m^*$, whereas~\cite{end-guarantees} only works for what is the equivalent of $\rho_m^*$. Related to this, the solution given in that work is very much tied to the optimal performance ratios, and the same holds for the optimality proof which is quite involved and does not extend in an obvious way to any $\rho$. The theoretical worst-case runtime of the algorithm in~\cite{end-guarantees} is $O(m^2 \log L)$, whereas the runtime of our algorithm has only an $O(m\log m)$ dependency on $m$, as guaranteed by 
Theorem~\ref{thm:complexity.star}. Given the conceptual similarities between the two problems, our techniques can be readily applicable to the scheduling problem as well, and provide the improvements we describe above. 

For clearance in networks, we demonstrated that RPT-based heuristics can have a significant impact on performance, in comparison to CPT-based heuristics. The RPT heuristic we implemented is from~\cite{frederickson1978approximation}, but more complex and sophisticated heuristics are known~\cite{corberan2010recent}. It would be interesting to further explore the impact of such heuristics in competitive search.

\bibliographystyle{plain}
\bibliography{targets,anytime,bijective,online}

\bigskip

\appendix
{\Large \bf Appendix}

\section{Formulating the LPs, and extendability}

We introduce the shorthand notation $S_j^{(X)} = \sum_{i=1}^j x_i$. When it is obvious which strategy we are referring to, we will simply use the notation $S_j$.

For the line and star environments, it is clear that we can restrict ourselves to strategies where each step has positive length, and  which go strictly further at each visit to a given ray. These conditions are implicit is our LP formulation.

\subsubsection{Competitiveness constraints}
It is known that the worst-case competitive ratio corresponds to targets placed immediately after the turn points, and thus it suffices to enforce $R$-competitiveness in those locations. So the total distance traveled by the searcher upon returning to a turn point for the first time must not exceed $R$ times the distance from the origin to this turn point. Using the notations we introduced at the beginning of the star section, we obtain:

\[ 2 S_{j-1} + x_{\bar{\jmath}} \leq (1+2\rho) x_{\bar{\jmath}} \Leftrightarrow S_{j-1}\leq \rho x_{\bar{\jmath}}, \]
which yields the constraint $(C_j)$.

When searching a new ray for the first time, say on step $i$, because we have assumed that the target is located at distance at least $1$ from the origin, we obtain the constraint $2S_{i-1}+x_{i}\leq 1+2\rho\Leftrightarrow S_{i-1}\leq \rho$. Obviously we only need to keep the last such constraint, corresponding to step $j_0$, which is the dominant constraint. Also, any competitiveness constraint before the step $j_0$ is superfluous, because the competitive factor is necessarily worse for points at the same distance but on ray $r_{j_0}$. We thus showed how to obtain constraint $(C_0)$. Constraint $(B)$ clearly reflects the budget requirement.

It remains to explain the extendability constraints. We do so in detail in what follows.

\subsubsection{Extendability constraints}

We begin with the line. 
As discussed in the main paper, in order to enforce the extendability property we consider targets placed just beyond the turn point at $x_{k-1}$, and just beyond the end point at $x_k$. For the end point $x_k$, this property is satisfied by the strategy: the searcher can visit a point hiding infinitesimally beyond $x_k$ at an infinitesimally small aditional cost, and without changing the competitive ratio.
 For the turn point at $x_{k-1}$, the extension of our strategy which gets there in the least time turns around at $x_k$, goes through $O$ and reaches the turn point at $x_{k-1}$,  and thus we get the following constraint:

\[ 2 S_k + x_{k-1} \leq (1+2\rho) x_{k-1} \Leftrightarrow S_k\leq \rho x_{k-1}. \]

For the star, the situation is analogous. 
For the end point $x_k$, as for the line, the property is trivially satisfied; for the other points, by considering extensions which turn around at $x_k$ to explore each other ray, we get the family of constraints

\[ 2 S_k + x_{l_r} \leq (1+2\rho) x_{l_r} \Leftrightarrow S_k\leq \rho x_{l_r}, \quad r\neq r_k. \]

In principle, we could apply this concepts in general environments, and we 
give the following formal definition:
\begin{definition}
Let $S$ be a finite search strategy on an environment $E$.
We denote $S(E)$ the part of the environment which is explored by $S$. We say that $S$ is $R$-extendable if for any point $P$ along the boundary of $S(E)$, there exist $S_P$ a strategy which extends $S$ (i.e. $S$ is a prefix of $S_P$) and $V_P$ a neighborhood of $P$ such that $V_P \subset S_P(E)$ and $\comp(S_P) \leq R$.
\end{definition}
In other words, an {\em $R$-extendable} strategy is an $R$-competitive strategy which can be extended to explore infinitesimally farther beyond any point on the boundary of the area explored up to time $T$, while keeping its competitive ratio below $R$.
Any prefix of an infinite strategy with bounded competitive ratio $R$ is $R$-extendable; in particular prefixes of the geometric and aggressive strategies are extendable.

\section{Computing the aggressive strategy on the star}

In this section we show that the aggressive strategy on the $m$-ray star is well-defined for any competitive ratio $R\leq R^*_m=1+2m^m/(m-1)^{m-1}$, and we give an explicit formula for it.

This {\em aggressive} strategy is a cyclic strategy $Z=(z_i)$ which successively maximizes the length searched at each step, within the competitive constraints. 
\cite{jaillet:online} show that this problem is well-defined, and that there is a strategy which satisfies the  linear recurrence relation
\[ z_{i+m}=\rho(z_{i+1}-z_i),\]
with $R=1+2\rho$. They give a ``canonical'' solution for optimal $R=R_m^*$, which we prove is the only solution to this recurrence; we also provide a formula for $R>R^*_m$ and prove its uniqueness.

As noted by \cite{jaillet:online}, there are two initial conditions that we can use to help determine the strategy:
\[ \sum_{i=1}^{m-1} z_i = \rho, \text{ and } \sum_{i=1}^m z_i = \rho z_1. \]
These correspond to the first two constraints which for finite strategies we denote $(C_0)$ and $(C_1)$, and all other constraints serve in the recurrence relationship, obtained by subtracting $(C_i)$ from $(C_{i+1})$. To our knowledge, no previous work has give an expression for $Z$, for general $R$ and $m\geq 3$, and in this section we show how to derive it.

The characteristic polynomial of the recurrence is $\chi(t)=t^m-\rho t+\rho$. If $\chi$ has a root $\zeta$ of order $n$ then $(a_0+a_1 i + \dots + a_{n-1} i^{n-1})\zeta^i$ is a solution to the recurrence, for any $(a_0,\dots,a_{n-1})\in \mathbb{C}^n$, and any solution is a linear combination of such terms.

By Descartes' rule of signs, $\chi$ has either two positive real roots (counting multiplicity) or none. Denote $r_m^*=\frac{m}{m-1}$. For $\rho\geq \rho_m^*=\frac{m^m}{(m-1)^{m-1}}$ we have $\chi(r_m^*)\leq 0$, so $\chi$ always has exactly two positive real roots, which we denote $\zeta_1$ and $\zeta_2$. For $\rho=\rho_m^*$, $\chi$ has a double root at $r^*_m$.

First we study the case when $\rho=\rho_m^*$. We can factor $\chi$:
\[\chi(t)=(t-r_m^*)^2(t^{m-2}+2r_m^*t^{m-3}+\cdots+(m-1){r_m^*}^{m-2})\]
\[ = \frac{t^{m-2}}{r_m^*}(t-r_m^*)^2\phi_m'\left(\frac{r^*_m}{t}\right), \]
where $\phi_m(t) = \frac{t^m-1}{t-1}$. $\phi_m$ has $m-1$ distinct roots on the unit circle, so all roots of $\phi'$ are distinct, and inside the convex hull of the roots of $\phi$, therefore of norm $<1$. This means that all roots of $\chi$ which are not $r^*_m$ are of norm $>r^*_m$, and as discussed above they must be negative or complex. Any meaningful solution to the recurrence formula must be positive, therefore these other roots cannot contribute to the solution. In conclusion, using the initialization constraints we obtain the following formula for $Z$:
\[ z_i = \frac{m+i-1}{m-1}\left(\frac{m}{m-1}\right)^i. \]

Now for the case when $\rho>\rho_m^*$. For $x\in ]\zeta_1,\zeta_2[$, we have $\chi(x)<0 \Leftrightarrow \rho x> \rho + x^m$,  and so Rouché's theorem tells us that there is exactly one root of norm $<|\zeta_2|$, which we know to be $\zeta_1$. Suppose that $\zeta_2 e^{i\theta}$ is a root of $\chi$. Then
\[ \zeta_2^m e^{mi\theta}-\rho\zeta_2 e^{i\theta} +\rho=0 \text{ and } \zeta_2^m = \rho(\zeta_2-1) \]
\[\Rightarrow \rho(\zeta_2-1)e^{im\theta} = \rho(\zeta_2 e^{i\theta}-1)
\Leftrightarrow \zeta_2 = \frac{e^{im\theta}-1}{e^{im\theta}-e^{i\theta}}\in \mathbb{R}\]
\[ \Rightarrow \frac{e^{im\theta}-1}{e^{im\theta}-e^{i\theta}}=\overline{\frac{e^{im\theta}-1}{e^{im\theta}-e^{i\theta}}}
=\frac{e^{im\theta}-1}{e^{im\theta}-e^{i\theta}}e^{i\theta} \Rightarrow e^{i\theta}=1. \]

This shows that $\zeta_2$ is the only root of $\chi$ of that norm, so all other roots are of norm $>|\zeta_2|$, and being negative or complex they cannot contribute to the solution. In conclusion, using the initialization constraints we obtain the following formula for $Z$:

\[ z_i = (1+\alpha) \zeta_2^i -\alpha \zeta_1^i , \text{ with } \alpha = \frac{\zeta_1(\zeta_2-1)}{\zeta_2-\zeta_1}. \]

Computing $\zeta_i$ can be done most efficiently with binary search using $1\leq \zeta_1 \leq r_m^* \leq \zeta_2 \leq \rho^{\frac{1}{m-1}}$.


\section{Cyclicality and monotonicity in $L_m^{(k)}$ }
\label{app:cyclic}

In this section we show that any optimal solution to $L_m^{(k)}$ corresponds to a cyclic and monotone strategy. The basic steps of the proof are similar to those found in~\cite{jaillet:online}.

We begin with a tightness lemma similar to Lemma~\ref{lemma:line.one.constraint}.
\begin{lemma}
In any optimal solution to $L_m^{(k)}$, at least one of $(C_0)$ and $(B)$ is tight. All other constraints $(C_j)$ and $(E_r)$ are tight.
\label{lemma:acyclic.star.one.constraint}
\end{lemma}
\begin{proof}
We extend the proof of Lemma \ref{lemma:line.one.constraint} to the case of the star. Suppose $X^*=(x^*_i)$ is an optimal solution to $L_m^{(k)}$, which does not satisfy the conditions of the lemma. Recall that there are implicit conditions in the formulation of $L_m^{(k)}$, namely $x_i>x_{\bar{\imath}}$.

If a constraint $(C_j)$ is not tight, then we can decrease $x^*_{\bar{\jmath}}$ by a small quantity $\delta$ and increase $x^*_k$ by $\delta$ in order to obtain a feasible solution with a higher objective value, which contradicts the optimality of $X^*$.

If $(C_0)$ and $(B)$ are both loose, then we can scale up $X^*$ by a factor $\alpha>1$, thus increasing the objective value, a contradiction.

Finally, if a constraint $(E_r)$ is loose, then decreasing $x_{l_r}$ and increasing $x_k$ by a small quantity $\delta$ creates a new feasible strategy which is also optimal, because it has the same objective. If $(C_k)$ exists, i.e. $j_0<k$, then constraint $(C_k)$ becomes loose, and if not, then constraints $(C_0)$ and $(B)$ become simultaneously loose; either case provides a contradiction to the above.
\end{proof}

The following property is very intuitive and will be needed to show cyclicality. Similar properties are very often useful in star search problems.

\begin{property}
Any optimal strategy visits, at each step, the ray which has been explored the least so far.
\label{property:star.ordering}
\end{property}
\begin{proof}
First, we prove, by way of contradiction, that any optimal strategy begins by visiting each ray once. Let $X^*$ be an optimal strategy, and recall that $j_0$ is the last step during which we explore a new ray.
Suppose that $X^*$ visits the same ray $r$ twice before step $j_0$, say at steps $i_1$ and $i_2$, with $i_2<j_0$. Then we could simply halve the size of $x_{i_1}$ and obtain a new feasible strategy $\hat{X}$ which has loose constraints: indeed, $x_{i_1}$ only shows up on the left-hand side of the inequalities in $L_m^{(k)}$, so all constraints are loosened. But by lemma \ref{lemma:acyclic.star.one.constraint} $\hat{X}$ cannot be optimal, and neither can $X^*$, which has the same objective value, a contradiction.

Now we look at the steps after $j_0$. From Lemma \ref{lemma:acyclic.star.one.constraint}, we get for any optimal strategy the set of equations $(S_{j-1}=\rho x_{\bar{\jmath}})$ and $S_k = \rho x_{l_r}, l_r\neq r_k$ (Recall the definition of $S$ that we gave in the first line of this Appendix). This makes it clear that $(x_{\bar{\jmath}})_{j_0\leq j\leq k}$ is an increasing series, and that the final steps on each ray are the last $m$ steps. This is precisely the statement of our lemma: at each step $i$, the length to which we had previously explored $r_i$ is increasing, or equivalently, at each step we visit the least explored ray. To see this more clearly, we give a proof by contradiction. Suppose there is a step $i_1$ where an optimal strategy $X^*$ visits a ray $r_1$ which has been explored more than the least explored ray $r_0$. $X^*$ can never return to visit $r_1$, because if it visits $r_0$ on step $i_0>i_1$, then $x^*_{\bar{\imath}_0}<x^*_{\bar{\imath}_1}$, a contradiction. But if $X^*$ never returns to ray $r_0$, then we have $x^*_{l_{r_0}}<x^*_{\bar{\imath}_1}$, a contradiction.
\end{proof}

Now we can move on to the main result.

\begin{theorem}
Any optimal solution to $L_m^{(k)}$ must be {\em monotone} and {\em cyclic}, that is $(x_i)$ must be increasing and $r_i = i \mod m$ up to a permutation.
\label{theorem:cyclic}
\end{theorem}
\begin{proof}
Let $X=(x_i, r_i)$ be an optimal solution to $L_m^{(k)}$.
The proof of monotonicity borrows the swapping idea from \cite{Gal80}.
Suppose that $X$ is not monotone, i.e. $\exists i_0, x_{i_0}>x_{i_0+1}$. Define strategy $Y=(y_i,s_i)$ to be a modification of strategy $(x_i,r_i)$ where we swap the two steps $x_{i_0}$ and $x_{i_0+1}$ as well as the roles that rays $r_{i_0}$ and $r_{i_0+1}$ play {\em after} the swap. Formally,
$y_i = x_i$ except for the swap $y_{i_0}=x_{i_0+1}, y_{i_0+1}=x_{i_0}$, and
$s_i = r_i$, except when  $i> i_0+1$ and we search $ r_{i_0}$ or $r_{i_0+1}$: in this case
$r_i=r_{i_0} \Rightarrow s_i=r_{i_0+1}$ and $r_i=r_{i_0+1} \Rightarrow s_i=r_{i_0}$.\\

Swapping does not increase the partial sums: $S_j^{(Y)} \leq S_j^{(X)}$ for all $j$, so $(C_0)$ holds for $Y$, as well as $(B)$. Swapping does not change the set of the last steps on each ray: $\{x_{l_r},l_r\}=\{y_{l_r},l_r\}$, and so if $i_0\neq k-1$, all constraints $(E_r)$ hold for $Y$. Most importantly, swapping has the nice property that for all $j$, $y_{\bar{\jmath}}=x_{\bar{\jmath}}$. So for all $j$, $(C_j)$ holds for $Y$. Recall the tightness property (lemma \ref{lemma:acyclic.star.one.constraint}).
\[ \rho y_{\overline{i_0+1}} = \rho x_{\overline{i_0+1}} = S_{i_0}^{(X)} > S_{i_0}^{(Y)} \]
so $(C_{i_0+1})$ is not tight for $Y$, therefore $Y$ cannot be optimal according to lemma \ref{lemma:acyclic.star.one.constraint}, and netiher can $X$, which has the same objective value: a contradiction.

We will address the case $i_0=k-1$ later. This does not impact the proof of cyclicality.


Now for cyclicality. We showed above that any optimal strategy must be monotone (up to step $k-1$). Recall property \ref{property:star.ordering}. Take an optimal strategy: we can suppose that it begins by visiting the rays in order, $1$ to $m$.
On step $m+1$, it needs to visit the least visited ray so far, which is ray $1$, because of monotonicity. Then ray $1$ becomes the ray which has been visited the most so far; an immediate induction follows, proving that the strategy is cyclic.


We left a piece of the monotonicity property hanging, the case where $i_0=k-1$: we still need to prove that if $X=(x^*_i)$ is an optimal strategy, then $x_k^*\geq x_{k-1}^*$. We can show this by applying algorithm \ref{alg:feasible.to.optimal} (see definition in the proof of lemma \ref{lemma:feasible.positive} on the next page) to the cyclic geometric strategy $G=(g_i)=({r_m^*}^i)$. We have the identity $g_k> g_{k-1}$ at the start of the algorithm, and at each step of the algorithm we increase $g_k$ and decrease $g_{k-1}$, before scaling up by a factor $\alpha$, hence $x_k^*\geq x_{k-1}^*$.

\end{proof}


\section{Finding the optimal values of $k$ for $P_m^{(k)}$}


In this section we prove Theorems \ref{theorem:star.optimal} and \ref{thm:complexity.star}. The main idea for the proof of Theorem \ref{theorem:star.optimal} is the same as for the line: we show that the terms in $X_0^{(k)}$ are increasing, therefore the largest feasible $k$ is optimal, and then we show that the objective values of $X_B^{(k)}$ are decreasing, therefore the smallest feasible $k$ is optimal. However, the proof is much more involved than the proof for the line, because $X_0^{(k)}$ is no longer simply a prefix of $X_0^{(k+1)}$. Lemma \ref{lemma:feasible.positive} is a technical result which allows us to prove Lemma \ref{lemma:theorem.part.one}, which results directly in the first part of the theorem; some more calculations give us the second half of the theorem in Lemma \ref{lemma:theorem.part.two}.

In this whole section, we discard all monotonicity constraints $(M_i)$, with the exception of the final one $x_k\geq x_{k-1}$, which we will relabel $(M)$. We also discard the implicit constraints $x_i>0$ and $x_{i+m}>x_i$, regarding $P_m^{(k)}$ as simply a set of equations.

The following technical result is key to efficiently determining the optimal values of $k$.

\begin{lemma}
Any point $X=(x_i)$ which is satisfies all of the constraints in $P_m^{(k)}$ is positive, that is for all $i$, $x_i\geq 0$.
Also, $x^*_k-x^*_{k-1} \geq x_k-x_{k-1}$.
\label{lemma:feasible.positive}
\end{lemma}
\begin{proof}
Take $X=(x_i)$ a feasible point for $P_m^{(k)}$. Using the methods from the proof of lemma \ref{lemma:acyclic.star.one.constraint}, we can transform $X$ into the optimal strategy $X^*$ by performing the process described in Algorithm \ref{alg:feasible.to.optimal}.

\begin{algorithm}[htb!]
\small
\DontPrintSemicolon
\small
\caption{Feasible to optimal}
\label{alg:feasible.to.optimal}
{\bf Input:} $X = (x_i)$ a feasible point for $P_m^{(k)}$ \;
\While {any constraint $(C_j)$ or $(E_j)$ is loose} {
\For {j=1,\dots, k-m} {
\If {$(C_j)$ is loose} {
select $\delta>0$ so that $(C_j)$ will become tight \;
$x_j, x_k \leftarrow x_j - \delta,\,\, x_k + \delta$ \;
}
}
\For {j=k-m+1,\dots,k-1} {
\If {$(E_j)$ is loose} {
select $\delta>0$ so that $(E_j)$ will become tight \;
$x_j, x_k \leftarrow x_j - \delta,\,\, x_k + \delta$ \;
}
}
}
select $\alpha$ such that $(C_0)$ or $(B)$ will become tight \;
$x_i \leftarrow \alpha x_i$ for all $i$ \;
{\bf Output:} $X=X^*$ the optimal solution \;
\end{algorithm}

This algorithm has a purely conceptual value, because every time we tighten a constraint, we loosen at least one other, and thus it cannot finish in finite time. However, convergence is guaranteed by the fact that $x_k$ increases at each iteration and is bounded from above by $T$, therefore it must converge. All other variables $x_i$ must also converge, because they are all decreasing, and each one has a total variation of less than $x^*_k-x_k$. The reason the output must be $X^*$ is that all constraints are tight, witht the exception of $(M)$ and at most one of $(C_0)$ and $(B)$.

By running algorithm \ref{alg:feasible.to.optimal} on $X$, we decrease each $x_i,\,i<k$ by a certain amount, then scale it up by some $\alpha\geq1$, and obtain $x^*_i\geq 0$, hence necessarily $x_i\geq 0$. Using constraint $(M)$ we see that $x_k\geq x_{k-1}\geq 0$.

We call attention to a subtle detail: without constraint $(M)$, we could have had $\alpha<0$, for example if we start from the negative version of the optimal solution $(-x^*_i)$.
But constraint $(M)$ cannot be tight in the optimal solution, as shown by executing the algorithm on the geometric strategy $G=({r^*_m}^i)$: constraint $(M)$ starts out being non-tight and loosens progressively as the algorithm runs, therefore it cannot be tight for $X^*$. Because constraint $(M)$ is satisfied for $(x_i)$ all throughout the process, we cannot have $\alpha<0$, which would flip $(M)$ and violate it in $X^*$, a contradiction.

Constraint $(M)$ can only get looser at each step of algorithm \ref{alg:feasible.to.optimal}, which proves the second part of our lemma.

\end{proof}

If we remove constraint $(B)$ from $P_m^{(k)}$, we get an LP $P_{m,0}^{(k)}$ for which the solution is $X_0^{(k)}$, and similarly, by removing constraint $(C_0)$ from $P_m^{(k)}$ we get an LP $P_{m,B}^{(k)}$, for which the solution is $X_B^{(k)}$. Lemma \ref{lemma:feasible.positive} and algorithm \ref{alg:feasible.to.optimal} can be readily extended to show that any feasible point for $P_{k,0}^m$ or $P_{k,B}^m$ is positive, and that the inequality corresponding to constraint $(M)$ is valid.

One would expect that as $k$ grows, giving $X_0^{(k)}$ more steps to explore the domain, it is able to explore farther; conversely, it seems reasonable, though not quite obvious, that once the time budget is used up, it is best to waste as little time as possible taking extra steps, which backtrack on previously covered ground, and so $X_B^{(k)}$ should perform best for smaller $k$. We will show that this is indeed the case.

\begin{lemma}
Denote $x_{0,i}^{(k)}$ the $i$-th step in the strategy $X_0^{(k)}$ for each $k$. For all $i$, $(x_{0,k-i}^{(k)})_{k\geq i}$ is increasing.
\label{lemma:theorem.part.one}
\end{lemma}
\begin{proof}
Fix $k$. In order to simplify notations, denote $X=(x_i)=X_0^{(k)}$ and $Y=(y_i) = X_0^{(k-1)}$. It suffices to show that $\forall i, x_i\geq y_{i-1}$.

First we need to work to prove the following inequality:
\begin{equation}
x_k-x_{k-1} \geq y_{k-1}-y_{k-2}.
\label{eq:delta.M}
\end{equation}
Set $r=m/(m-1)$. Recall the characteristic polynomial $p(t) = t^m-\rho t +\rho$. For the optimal $\rho^*_m$, we have $p(r)=0 \Rightarrow r^m/(r-1) = \rho^*_m$, so in the general case $r^m/(r-1) \leq \rho$. This gives us the following identity:
\[ \sum_{i=0}^{j+m-1} (r-1)r^i = r^{j+m} - 1 \leq \rho (r-1)r^j - 1. \]
Define $Z = (z_i)_{i\leq k}$ by $z_i = y_{i-1} + (r-1)r^{i-1}$, using the convention $y_0=1$. $Z$ is a feasible point for $P_{k,0}^m$. Indeed, we verfiy each constraint:
\begin{align*}
S_{m-1}^{(Z)} &= S_{m-2}^{(Y)} + 1 +  \sum_{i=0}^{m-2} (r-1)r^i \\
&\leq \rho - y_{m-1} + \frac{\rho}{m} \leq \rho, & (C_0)
\end{align*}
because $(r-1)/r = 1/m$, and due to monotonicity (Theorem \ref{theorem:cyclic}), step $y_{m-1}$ needs to account for at least $1/(m-1)$ of the sum $S_{m-1}=\rho$, so $y_{m-1}\geq \rho/(m-1) \geq \rho/m$.
\begin{align*}
S_{j+m-1}^{(Z)} &= S_{j+m-2}^{(Y)} + 1 +  \sum_{i=0}^{j+m-2} (r-1)r^i \\
&\leq \rho y_{j-1} + \rho(r-1)r^{j-1} = \rho z_j, & (C_j)
\end{align*}
and similarly each $(E_j)$ holds. Constraint $(M)$ also holds:\\
$z_k = y_{k-1}+(r-1)r^{k-1}\geq y_{k-2} + (r-1)r^{k-2} = z_{k-1}$.
We finish by applying the second half of lemma \ref{lemma:feasible.positive} to $Z$: 
\begin{align*}
x_k - x_{k-1} &\geq z_k-z_{k-1} = y_{k-1}-y_{k-2} + (r-1)^2 r^{k-2}\\
&\geq y_{k-1} - y_{k-2}.
\end{align*}

Now that we have \eqref{eq:delta.M}, we can apply lemma \ref{lemma:feasible.positive} to the difference of the strategies $X$ and $Y$, in order to show that $X$ is ``bigger'' than $Y$.
Define $\Delta = (\delta_i)_{i\leq k}$ by $\delta_i = x_i - y_{i-1}$, with the convention $y_0=1$. $\Delta$ is a feasible point for $P_{k,0}^m$. Indeed, we can verify each constraint:
\[ S_{m-1}^{(\Delta)} = S_{m-1}^{(X)}-S_{m-2}^{(Y)} = y_{m-1} \leq \rho,  \qquad(C_0) \]
\[ S_{j+m-1}^{(\Delta)} = \rho x_j - \rho y_{j-1} = \rho \delta_{j}, \quad  j \in [1,k-m+1] \, (C_j) \]
\[ S_{k}^{(\Delta)} = \rho x_j - \rho y_{j-1} = \rho \delta_{j}, \quad j \in [k-m+1,k-1] \, (E_j) \]
\[ \text{and finally } \eqref{eq:delta.M} \Leftrightarrow \delta_k \geq \delta_{k-1}. \qquad (M) \]

Using lemma \ref{lemma:feasible.positive} we obtain that for all $i$, $\delta_i \geq 0 \Leftrightarrow x_i\geq y_{i-1}$, which concludes our proof.
\end{proof}

\begin{corollary}
Both the total length cleared and the time taken by strategy $X_0^{(k)}$ are increasing.
\end{corollary}
\begin{proof}
The objective is $\sum_{i=0}^{m-1} x_{0,k-i}^{(k)}$ which is a sum of increasing series; so is the time taken $S_{k-1} + S_k$.
\end{proof}
\begin{corollary}[First half of Theorem \ref{theorem:star.optimal}]
There is a critical value $k_0$ such that $X_0^{(k)}$ is feasible for $P_m^{(k)}$ if and only if $k\leq k_0$. This critical value achieves the maximum clearance among all feasible strategies $X_0^{(k)}$.
\label{cor:half.optimal}
\end{corollary}

For $X_B^{(k)}$, we have the reverse situation, where the lowest feasible $k$ yields the optimal solution. The proof is a bit more difficult.

\begin{lemma}[Second half of Theorem \ref{theorem:star.optimal}]
There is a critical value $k_B$ such that $X_B^{(k)}$ is feasible for $P_m^{(k)}$ if and only if $k\geq k_B$; either $k_B=k_0$ or $k_B=k_0+1$. The optimal strategy among all feasible $X_B^{(k)}$ is $X_B^{(k_B)}$.
\label{lemma:theorem.part.two}
\end{lemma}
\begin{proof}
First, because $X_0^{(k)}$ and $X_B^{(k)}$ both belong to the same line $\Delta_m^{(k)}$, they are scaled versions of each other. We saw that the time taken by $X_0^{(k)}$ increases with $k$, until constraint $(B)$ is surpassed for $k>k_0$. Before this point, $X_B^{(k)}$ is infeasible for $P_m^{(k)}$ due to constraint $(C_0)$. If $(B)$ is tight for $X_0^{(k_0)}$, then $k_B=k_0$ and the two strategies $X_0^{(k_0)}$ and $X_B^{(k_B)}$ are identical. If not, then $k_B = k_0+1$.

Now we show that $X_B^{(k_B)}$ is optimal. In order to simplify notations, denote $X=(x_i)_{i\leq k} = X_0^{(k)}$ and $Y=(y_i)_{i\leq k-1} = X_0^{(k-1)}$. Note that $(\gamma x_i)_{i\geq 2}$ is a feasible point for $P_m^{(k-1)}$, for suitably small $\gamma \leq 1$, chosen to make constraint $(C_0)$ hold. Apply algorithm \ref{alg:feasible.to.optimal} to $(\gamma x_i)_{i\geq 2}$, and denote $X^{\#}=(x_i^{\#})_{i\geq 2}$ the value taken by our strategy right before we scale it up by $\alpha$, i.e. the value of $X$ if we halt the algorithm at line 16. Considering constraint $(M)$, which only gets looser as the algorithm runs, we have the following identity:
\[ \gamma(x_{k}-x_{k-1}) \leq x_k^{\#}-x_{k-1}^{\#} = (y_{k-1}-y_{k-2})/\alpha, \]
and considering the penultimate step of our strategy:
\[ \gamma x_{k-1} \geq x_{k-1}^{\#} = y_{k-2}/\alpha. \]

Dividing these two identities by each other, we obtain a key inequality:

\begin{equation}
\frac{x_k-x_{k-1}}{x_{k-1}} \leq \frac{y_{k-1}-y_{k-2}}{y_{k-2}} \Leftrightarrow \frac{x_k}{x_{k-1}} \leq \frac{y_{k-1}}{y_{k-2}}.
\label{eq:key}
\end{equation}

Denote the total area cleared by strategy $X$ by
\[ {\bf clr}(X) = \frac{x_{k-m+1}+\dots+x_{k}}{2S_{k-1}^{(X)}+x+k}T=\frac{(m-1)x_{k-1}+x_k}{2\rho x_{k-1}-x_k}T. \]

We conclude by showing ${\bf clr}(X) \leq {\bf clr}(Y)$:

\[ \frac{(m-1)x_{k-1}+x_k}{2\rho x_{k-1}-x_k} \leq \frac{(m-1)y_{k-2}+y_{k-1}}{2\rho y_{k-2}-y_{k-1}} \]
\[ \Leftrightarrow (2\rho + m - 1) x_k y_{k-2} \leq (2\rho + m - 1) x_{k-1} y_{k-1} \]
(developing the cross-product and removing identical terms)
\[ \Leftrightarrow \frac{x_k}{x_{k-1}} \leq \frac{y_{k-1}}{y_{k-2}} \text{ which is the inequality \eqref{eq:key}.}\]
\end{proof}

%

Lemma \ref{lemma:theorem.part.two} also provides the missing detail (the fact that $k_0$ is close to $k_B$) needed to complete the proof of Theorem \ref{thm:complexity.star}.

\section{Further experimental results}

\subsection{Implementation details}

We implemented the algorithms for both the star and the network in 
Python, and we run the experiments on a standard laptop. We implemented 
\china and \rural using the NetworkX library 
(https://networkx.github.io).

As stated in the main paper, we used networks from the online library 
{\em Transportation Network Test Problems}~\cite{transportation:2002}. 
We made the following minor modifications: we made the networks 
undirected, contracted nodes joined by edges of length $0$, and then 
scaled each network so that the shortest edge has length $4$: the last 
step is necessary because some networks have lengths in miles and others 
in meters.
Table~\ref{table:city.specs} shows the sizes of the networks we used for 
our experiments.

\begin{table}[htb!]
\centering
\begin{tabular}{c|c|c}
Network        &Nodes        &Edges\\
\hline
Sioux Falls        &24        &38\\
Eastern Massachussets&        74&129\\
Friedrichshain        &144        &240\\
Berlin                &633        &1042\\
Chicago        &933        &1475
\end{tabular}
\caption{Sizes of the networks used in our experiments.}
\label{table:city.specs}
\end{table}

\subsection{Experiments on the star}

We observed that our optimal strategy has a strong relative advantage 
over the other two strategies (the mixed aggressive and the scaled 
geometric). Table~\ref{table:star.m.rho} demonstrates this advantage, 
for different values of $m$ and $R$, and for a budget $T$ fixed to 
$T=10^{16}$. As shown in Figure~\ref{fig:star.gain}, for
smaller $T$ we expect an even stronger advantage of the optimal 
strategy. From the same figure, we observe that
as $T$ becomes even larger than $10^{16}$, we expect the same asymptotic 
behavior as shown in Table~\ref{table:star.m.rho}.  The relative 
advantage reaches $42\%$ for large values of $m$, and is significant for 
a wide range of values of $R$.
For much larger values of $R$ (i.e. $R\geq 100 R_m^*$), 
the relative advantage does eventually drop to 1, at which point the 
strategies are practically indistinguishable in terms of clearance. 


  \begin{table}
  \centering
  \begin{tabular}{c|c|c|c|c}
  \diagbox{$m$}{$\frac{R}{R^*_m}$}        & 1 & 2&5&10\\
  \hline
3	&1.124	&1.156	&1.126	&1.100	\\
4	&1.197	&1.266	&1.240	&1.205	\\
5	&1.244	&1.342	&1.329	&1.294	\\
10	&1.335	&1.521	&1.562	&1.550	\\
20	&1.384	&1.625	&1.712	&1.726	\\
50	&1.413	&1.692	&1.814	&1.850	\\
100	&1.424	&1.715	&1.850	&1.894
  \end{tabular}
  \caption{Relative advantage of the optimal strategy over the other two strategies, for various values of $m$ and 
$R$. Each entry is the ratio of the clearance achieved by the optimal
strategy over the clearance of the scaled aggressive strategy}
  \label{table:star.m.rho}
  \end{table}

\subsection{Experiments on networks}

We found that in practice, \rural never took longer than \china to 
complete a tour: this is in part due to the fact that \rural performs 
its tour on a smaller subgraph than \china. We also added a small 
variation to \rural: we do not require the RPT to return to the origin, 
and once all edges have been traversed, we use the current node as the 
starting point for the next tour.

We present further experiments showing runs for other networks in our 
dataset (Figures~\ref{fig:EMA} and~\ref{fig:frd}). We see that \rural performs better than \china even for smaller 
networks, though the results are more pronounced for the larger ones, as 
expected.

\begin{figure}[htb!]
\centering
\includegraphics[width=0.65\linewidth]{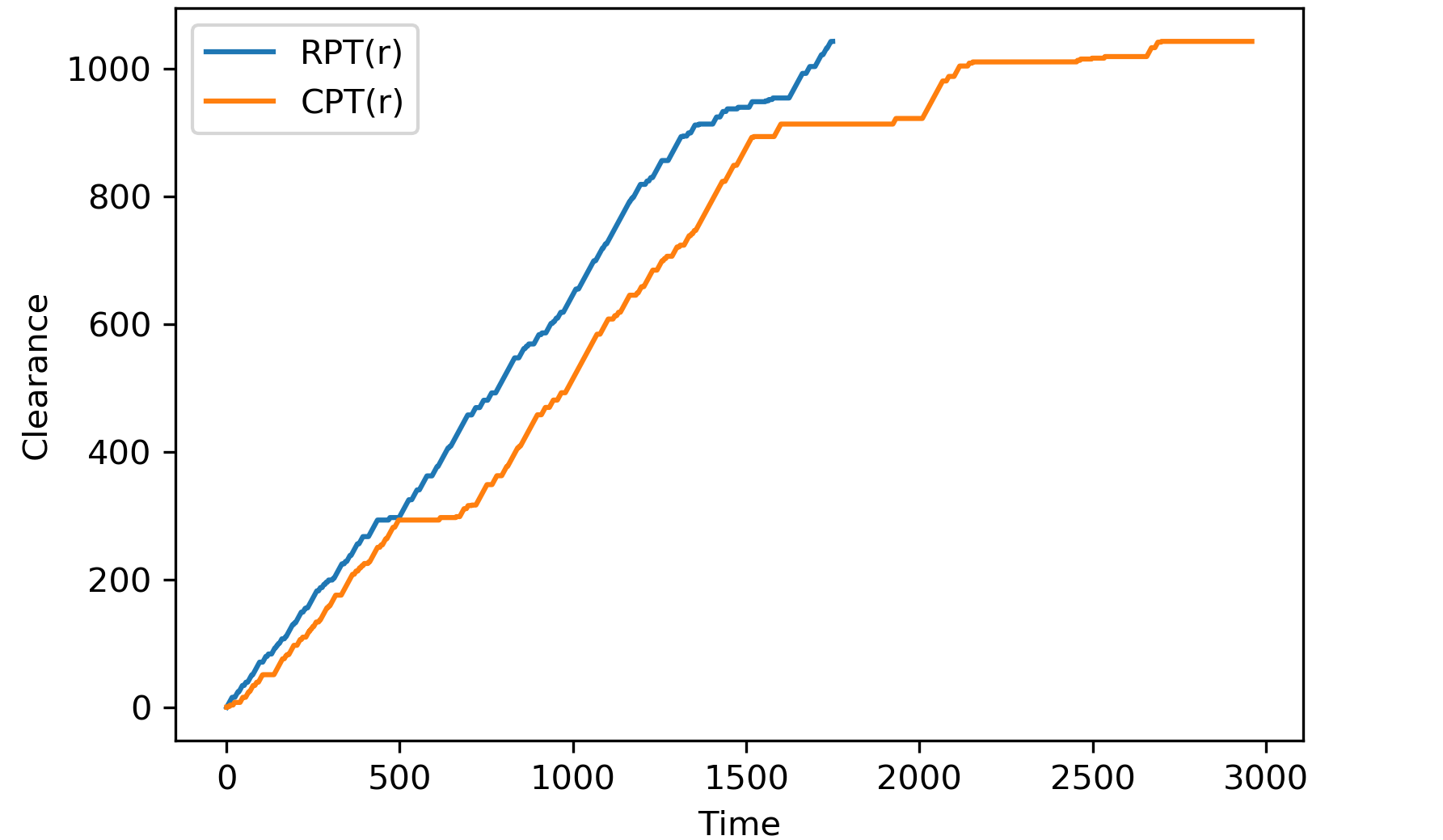}
\caption{Comparison of the two strategies on the Eastern-Massachussets network (74 nodes, 129 edges).}
\label{fig:EMA}
\end{figure}

\begin{figure}[htb!]
\centering
\includegraphics[width=0.65\linewidth]{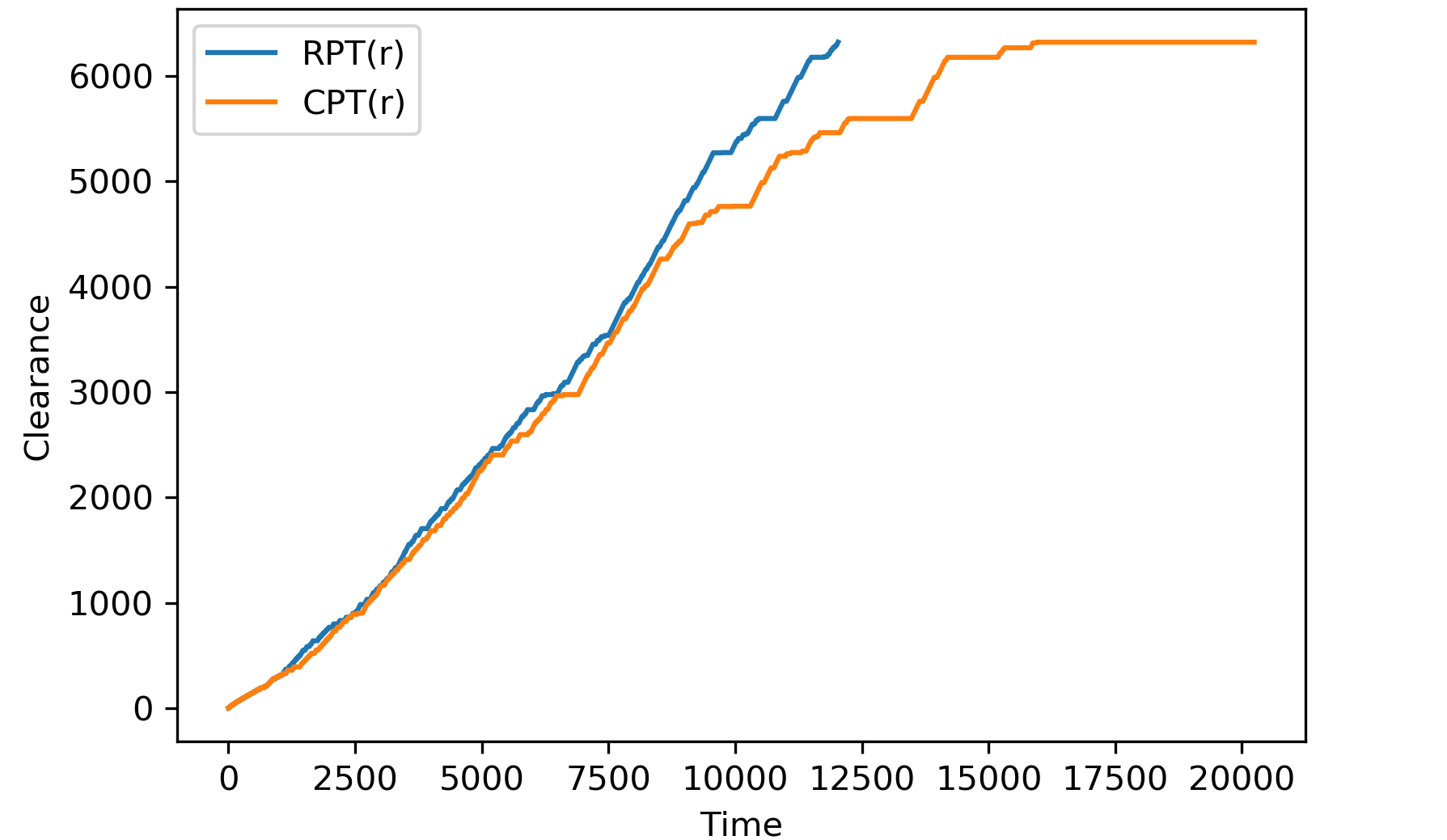}
\caption{Comparison of the two strategies on the Friedrichshain network (144 nodes, 240 edges).}
\label{fig:frd}
\end{figure}

Figure~\ref{fig:compr} depicts the influence of the parameter $r$ on the 
competitive ratios of \rural and \china, as run on the small Sioux Falls network, starting from a node located near the center of the network. We observe that there is indeed 
a minimum competitive ratio reached for $r\approx2$. 
Interestingly, this is in accordance with Proposition~\ref{prop:rural}, 
which shows that choosing $r=2$ yields the best approximation to the 
competitive ratio, for both \rural and \china.

\begin{figure}[htb!]
\centering
\includegraphics[width=0.65\linewidth]{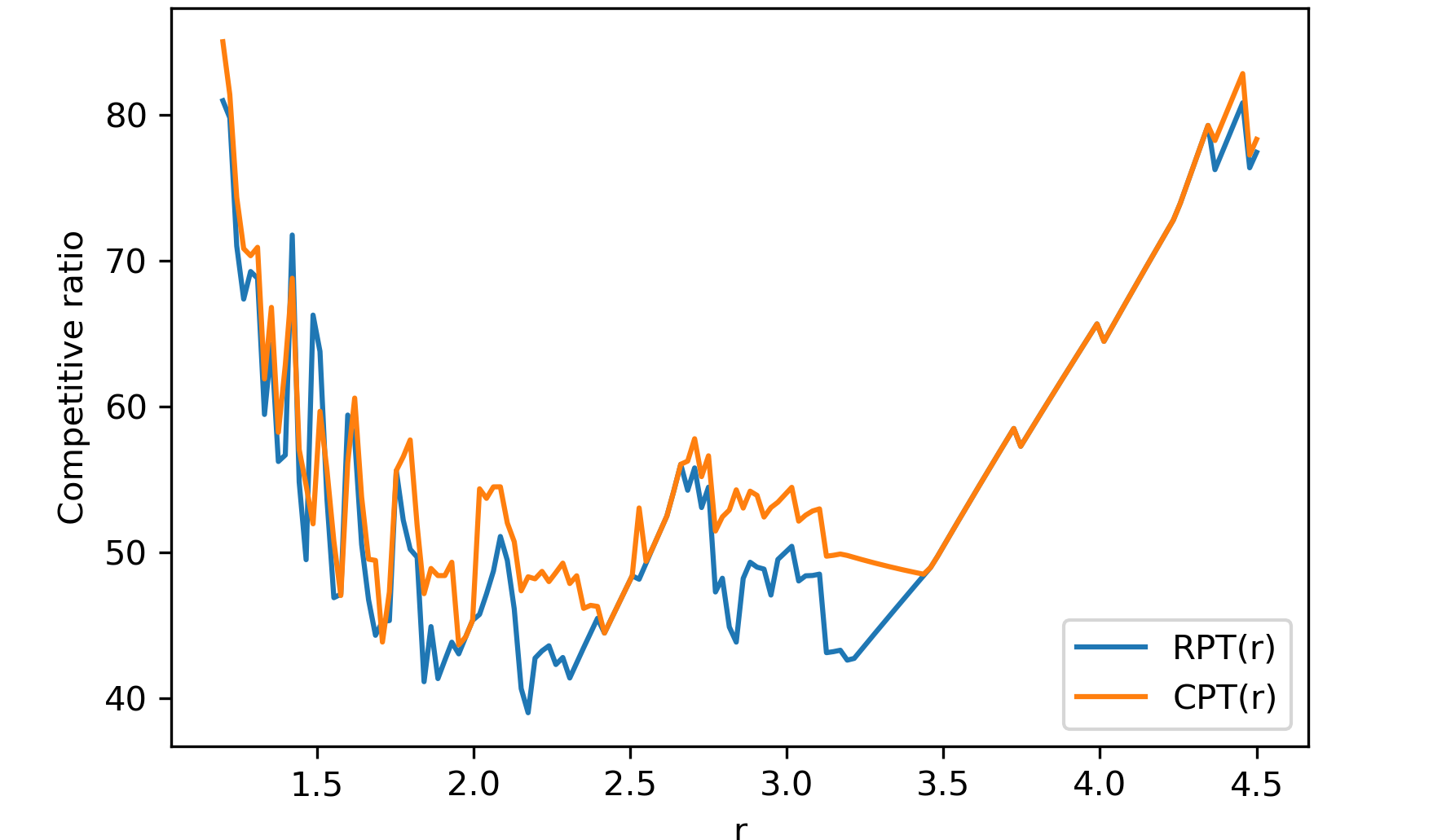}
\caption{Competitive ratio of the strategies as a function of the parameter $r$, calculated on the Sioux Falls network (24 nodes, 38 edges).}
\label{fig:compr}
\end{figure}

Figure \ref{fig:compchicago} is analogous to Figure \ref{fig:compare}, but for 45 random runs on the Chicago network. We see that the relative advantage of \rural over \china is even greater for a larger network.

\begin{figure}[htb!]
\centering
\includegraphics[width=0.65\linewidth]{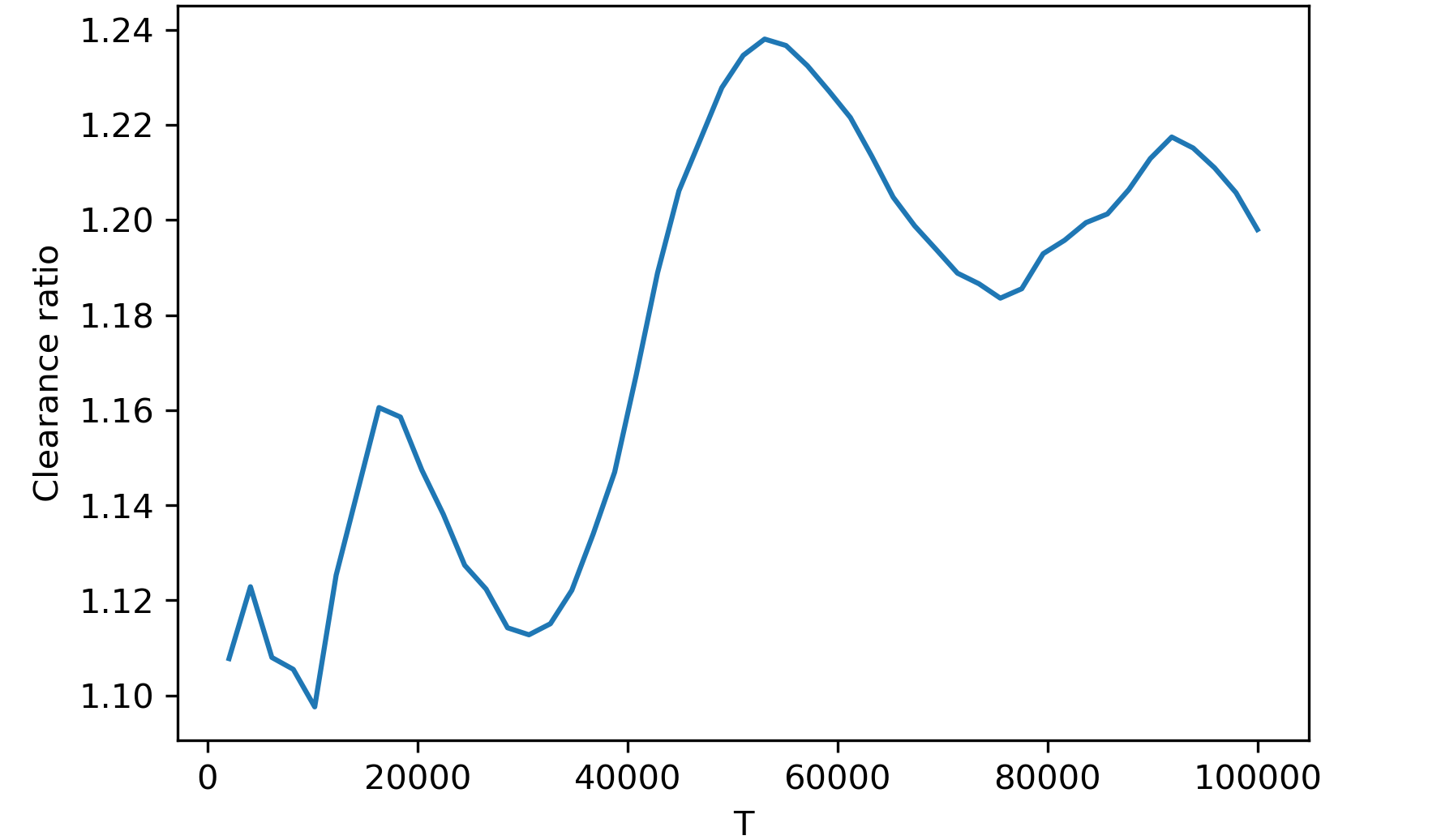}
\caption{Clearance ratio of \rural versus \china, for 45 randomly chosen roots, for the Chicago network.}
\label{fig:compchicago}
\end{figure}

Figure \ref{fig:rchicago} depicts the competitive ratios of each strategy over those 45 runs, sorted by increasing competitive ratio for \rural. We see that \rural is consistently much more efficient than \china, and it is also much more stable, especially for those roots for which the algorithms yield larger competitive ratios. 
The average competitive ratio over these runs for \rural is 152, compared to 200 for \china; the standard deviations are 26 and 39 respectively.

\begin{figure}[htb!]
\centering
\includegraphics[width=0.65\linewidth]{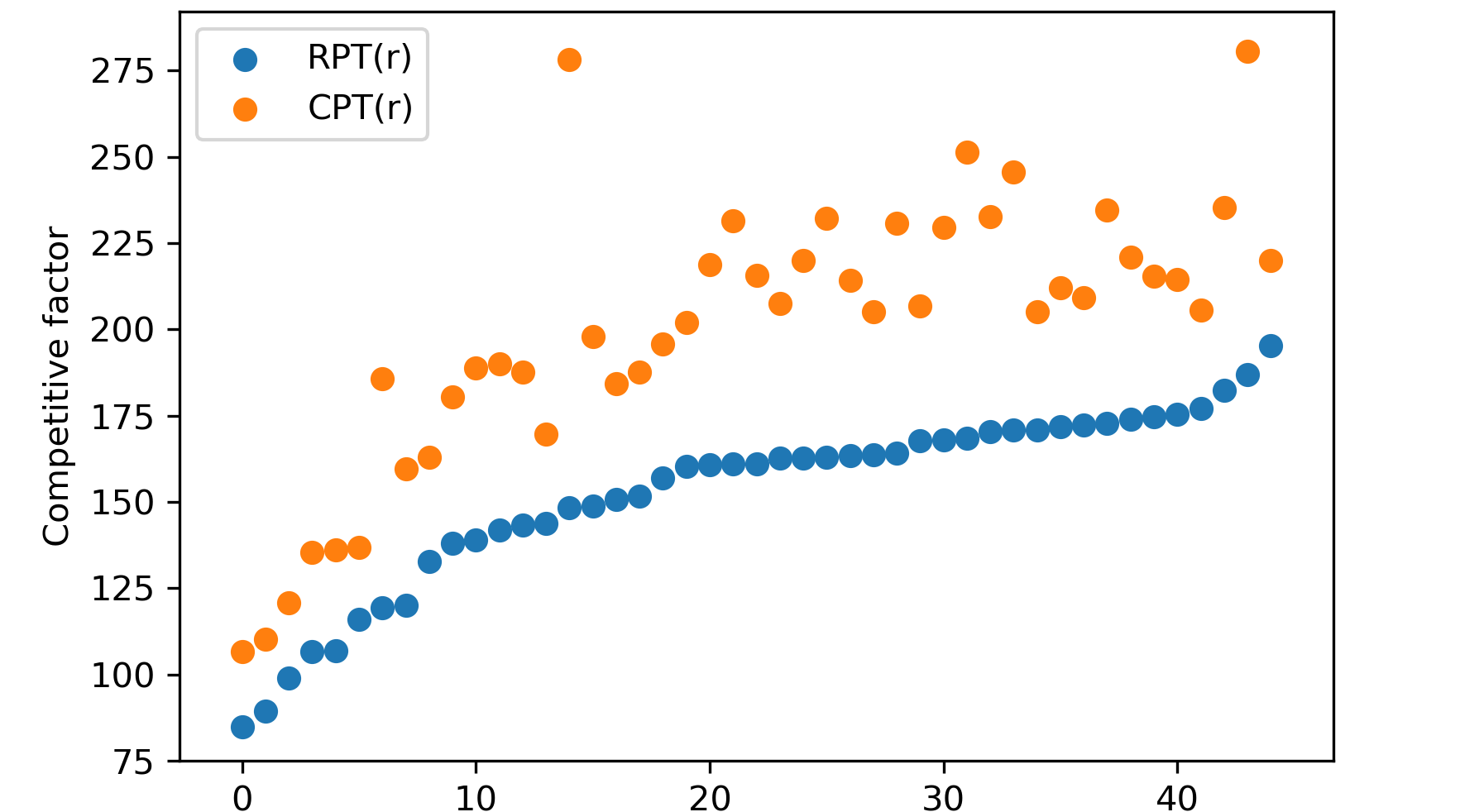}
\caption{Competitive ratios of the strategies for 45 random chosen roots, for the Chicago network.}
\label{fig:rchicago}
\end{figure}

\section{Solving the Earliest Clearance problem}

We give an overview about how the techniques we used in the context of the Maximum Clearance problem can help us solve this ``dual'' online search problem. Recall that the problem is defined in the last section of the main paper.

\subsubsection{The line environment}
For the unbounded line, we have an LP formulation similar to $L_2^{(k)}$, where we ``exchange'' the objective function and the final constraint: namely, we want to {\em minimize} $2\sum_{i=1}^{k-1}+x_k$, and add the constraint $x_k+x_{k-1}\geq L$. We can prove that in an optimal solution, all but one constraints must be tight, similarly to Lemma~\ref{lemma:line.one.constraint}, though for this problem only the first constraint $(C_0)$ may be loose. We can argue that the scaled aggressive strategy is optimal, since the final constraint $x_k+x+{k-1}\geq L$ is always tight.

\subsubsection{The star environment}
We can formulate this problem using an LP similar to $L_m^{(k)}$. First we can show a tightness result similar to Lemma \ref{lemma:star.one.constraint}, though this problem is easier: all constraints are tight except for possibly $(C_0)$. The proof of monotonicity and cyclicality is identical. This allows us to consider the LP in cyclic form:

\begin{align*}
\text{min}& \quad  2\sum \nolimits_{i=1}^{k-1} x_i+ x_k &\tag{$P_m^{(k)}$}\\
\text{subj to}& \quad \sum \nolimits_{i=1}^{m-1} x_i\leq \rho &\tag{$C_0$} \\ 
&\sum \nolimits_{i=1}^{j+m-1} \!x_i\leq \rho \cdot x_{j},   \qquad\,	  j \in [1, k-m] &\tag{$C_j$}\\
&\!\!\!\!\sum \nolimits_{i=1}^k x_i \leq \rho\cdot x_j,  	\;\;	j \in [k-m+1,k-1]  &\tag{$E_j$} \\
&x_i \leq x_{i+1},	\qquad\qquad\qquad\quad\;	  i \in [1,k-1]&\tag{$M_i$}\\
&\sum \nolimits_{i=0}^{m-1} x_{k-i} \geq L    &\tag{$D$}
\end{align*}

Each $P_m^{(k)}$ has a single solution which can be obtained in time $O(k)$ using Gaussian elimination on a matrix equation similar to \eqref{equation:the.matrix}. We can show by the same methods used in the proof of Theorem \ref{theorem:star.optimal} that the feasible solution with the fewest steps is the optimal solution, and with the geometric strategy giving an upper bound on this number of steps, we can use binary search to find the solution in time $O(m\log T \log(m\log T))$.

The experimental results we observe are extremely similar to those for Maximum Clearance: in short, the optimal strategy dominates the scaled aggressive and geometric strategies, and the same dependencies on $m$ and $R$ are observed.

\subsubsection{General networks}
For general networks, we use the same heuristic as for the Maximum Clearance problem: specifically, we run \rural until a total length $L$ has been cleared, using $r=2$. Similar conclusions can be reached, and we can quantify the relative improvement of \rural over \china. For example, from Figure~\ref{fig:frd} we can deduce for each value of clearance $L$, the time it took the two heuristics to clear length $L$.

Figure~\ref{fig:ecchicago} is analogous to Figure~\ref{fig:compchicago}, and depicts the average ratio between the time taken by \china and the time taken by \rural as a function of the desired length $L$. We observe the expected improvements, which get quite significant for large values of $L$.

\begin{figure}[htb!]
\centering
\includegraphics[width=0.65\linewidth]{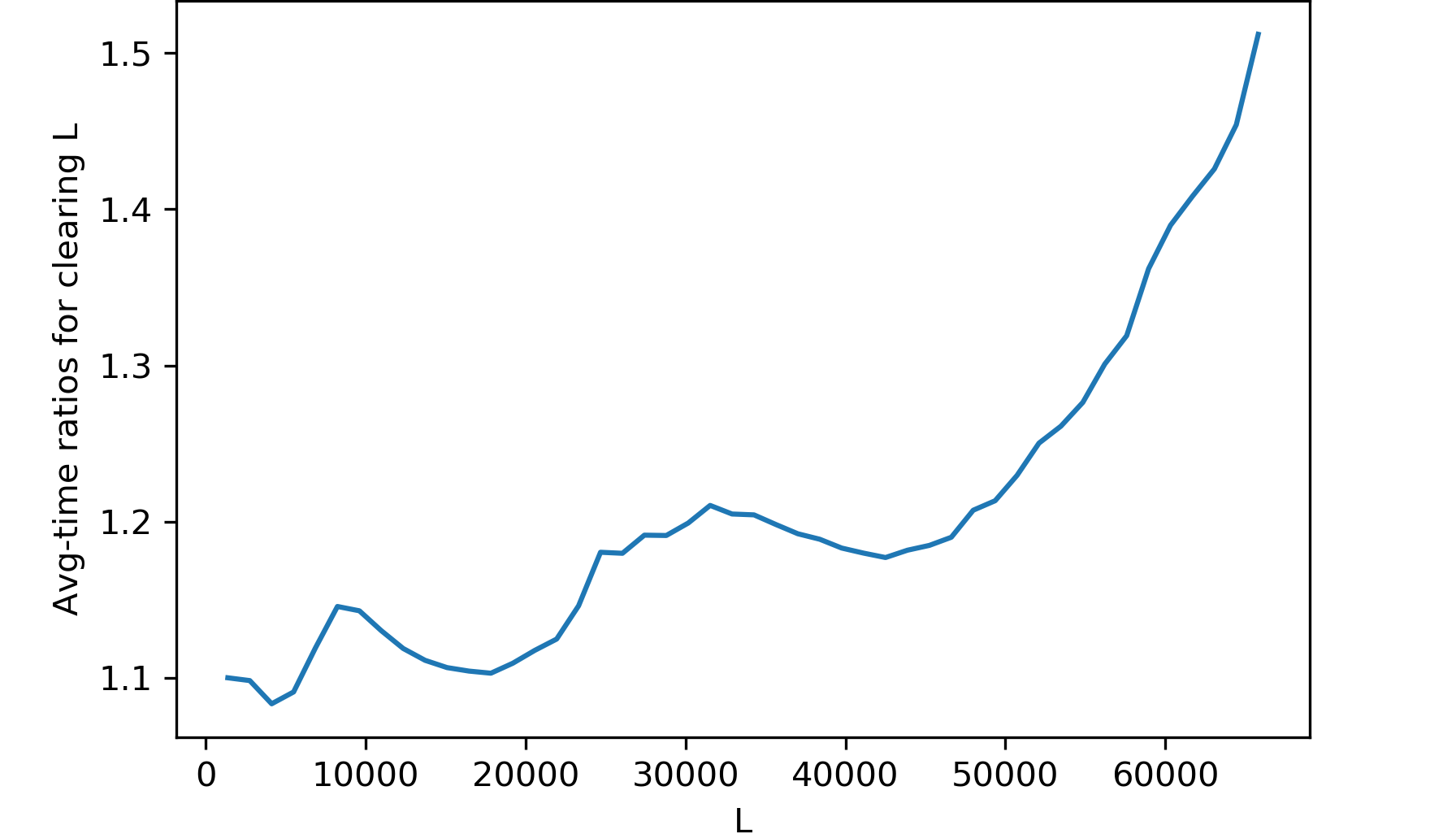}
\caption{Speed of \rural versus \china, for 45 randomly chosen roots, for the Chicago network.}
\label{fig:ecchicago}
\end{figure}

\end{document}